\documentclass[letterpaper]{statsoc}

\pdfoutput=1

\usepackage{amssymb}
\usepackage{amsmath,bm}
\usepackage{graphicx}   
\usepackage{verbatim}   
\usepackage{multirow}
\usepackage{enumitem}
\usepackage{geometry}
\usepackage{setspace}

\usepackage{color}
\usepackage{ulem}

\definecolor{orange}{rgb}{1,0.5,0}

\def\diag{\qopname\relax o{diag}}

\def\argmin{\qopname\relax m{argmin}}
\def\argmax{\qopname\relax m{argmax}}

\def\diag{\mathrm{diag}}
\def\sign{\mathrm{sign}}

\newtheorem{thm}{Theorem}
\newtheorem{lem}{Lemma}
\newtheorem{pro}{Proposition}

\newtheorem{rem}{Remark}

\def\wpi_0{\widetilde{\pi_0}}

\def\bzero{\mathbf 0}
\def\b1{\mathbf 1}

\def\bd{\mathbf d}
\def\be{\mathbf e}
\def\bf{\mathbf f}
\def\bg{\mathbf g}

\def\bv{\mathbf v}
\def\bw{\mathbf w}
\def\bx{\mathbf x}
\def\by{\mathbf y}

\def\cA{\mathcal{A}}
\def\cB{\mathcal{B}}
\def\cC{\mathcal{C}}
\def\cD{\mathcal{D}}

\def\cG{\mathcal{G}}

\def\cIG{\mathcal{IG}}

\def\cN{\mathcal{N}}

\def\balpha{\boldsymbol \alpha}
\def\bbeta{\boldsymbol \beta}

\def\bepsilon{\boldsymbol \epsilon}
\def\bgamma{\boldsymbol \gamma}
\def\bkappa{\boldsymbol \kappa}
\def\bxi{\boldsymbol \xi}
\def\bmu{\boldsymbol \mu}
\def\blambda{\boldsymbol \lambda}

\def\bomega{\boldsymbol \omega}
\def\btheta{\boldsymbol \theta}

\def\itOmega{\mathit \Omega}
\def\itSigma{\mathit \Sigma}

\title{Scalable Bayesian Variable Selection for Structured High-dimensional Data}

\author{Changgee Chang, Suprateek Kundu and Qi Long}
\address{Department of Biostatistics and Bioinformatics, Emory University, Atlanta, USA}

\begin{document}

\begin{abstract}
Variable selection for structured covariates lying on an underlying known graph is a problem motivated by practical applications, and has been a topic of increasing interest. However, most of the existing methods may not be scalable to high dimensional settings involving tens of thousands of variables lying on known pathways such as the case in genomics studies. We propose an adaptive Bayesian shrinkage approach which incorporates prior network information by smoothing the  shrinkage parameters for connected variables in the graph, so that the corresponding coefficients have a similar degree of shrinkage. We fit our model via a computationally efficient expectation maximization algorithm which scalable to high dimensional settings ($p {\sim} 100{,}000$). Theoretical properties for fixed as well as increasing dimensions are established, even when the number of variables increases faster than the sample size. We demonstrate the advantages of our approach in terms of variable selection, prediction, and computational scalability via a simulation study, and apply the method to a cancer genomics study.
\end{abstract}

\footnote{{\noindent \em Corresponding Author}: Suprateek Kundu, Department of Biostatistics \& Bioinformatics, Emory University, 1518 Clifton Road, Atlanta, Georgia 30322, U.S.A. \\ {\noindent \em Email}: suprateek.kundu@emory.edu }

{\noindent Keywords:} adaptive Bayesian shrinkage; EM algorithm; oracle property; selection consistency; structured high-dimensional variable selection.\\

\section{Introduction}

With the advent of modern technology such as microarray analysis and next generation sequencing in genomics, recent studies rely on increasingly large amounts of data containing tens of thousands of variables. For example, in genomics studies, it is common to collect gene expressions from $p\sim 20{,}000$ genes, which is often considerably larger than the number of subjects in these studies, resulting in a classical small $n$, large $p$, problem. In addition, it is well-known that genes lie on a graph of pathways where nodes represent genes and edges represent functional interactions between genes and gene products. Currently, there exist several biological databases which store gene network information from previous studies \citep{Stingo2011}, and these databases are constantly updated and augmented with newly emerging knowledge.

In such cases when genes are known to lie on an underlying graph, usual variable selection approaches such as Lasso \citep{Tibshirani1996}, adaptive Lasso \citep{Zou2006}, or spike and slab methods \citep{Mitchell1988} may run into difficulties, since they do not exploit the association structure between variables which may give rise to correlated predictors. Moreover, there is increasing evidence that incorporating prior graph information, where applicable, can improve prediction and variable selection in analysis of high dimensional data. \citet{Li2008} and \citet{Pan2010} proposed network-based penalties in linear regression, which induce sparsity of estimated effects while encouraging similar effects for connected variables. In a Bayesian framework, \citet{Li2010}, \citet{Stingo2011a}, and \citet{Stingo2011}, used spike and slab type priors for variable selection and Markov random field (MRF) type priors on variable inclusion indicators to incorporate graph information. More recently, \citet{Rockova2014} proposed an expectation maximization (EM) algorithm for variable selection using spike and slab priors which is known as EMVS and extended EMVS to incorporate graph information via MRF priors where a variational approximation was used in computation.
\citet{Rockova2014a} proposed a normal-exponential-gamma shrinkage approach with incorporation of the pathway membership information and developed an EM algorithm for computation.

To our knowledge, there is a scarcity of scalable Bayesian approaches for structured variable selection that possess desirable theoretical and numerical properties in high dimensions. The Bayesian approaches involving MRF type priors are implemented using Markov chain Monte Carlo and hence are not scalable to high dimensions involving tens of thousands of variables, such as in our cancer genomics application. While the EM approach by \citet{Rockova2014a} can incorporate pathway membership information, it is not equipped to incorporate edge information which is the focus of this article. Moreover, the theoretical properties and scalability of their method to the higher dimensions considered in this work ($p\sim 100{,}000$) are unclear. 
The variational approximation proposed by \citet{Rockova2014} may suffer from the loss of convexity properties and inferior estimates close to the transition points for tuning parameters, as indicated by the authors. 
The frequentist network-based regularization approaches are expected to be more scalable, but make a strong assumption of smoothness of covariate effects for connected variables in the graph, which may be restrictive in real-life applications. 

We propose a Bayesian shrinkage approach and an associated EM algorithm for structured covariates, which is scalable to high dimensional settings and possesses a desirable oracle property in variable selection and estimation for both fixed and increasing dimensions.  The proposed approach assigns Laplace priors to the regression coefficients and incorporates the underlying graph information via a hyper-prior for the shrinkage parameters in the Laplace priors. Specifically, the shrinkage parameters are assigned a log-normal prior specifying the inverse covariance matrix as a graph Laplacian \citep{Chung1997,Ng2002}, which has a zero or positive partial correlation depending on whether the corresponding edge is absent or present. This enables smoothing of shrinkage parameters for connected variables in the graph and conditional independence between shrinkage parameters for disconnected variables. Thus, the resulting approach encourages connected variables to have a similar degree of shrinkage in the model without forcing their regression coefficients to be similar in magnitude. The operating characteristics of the approach can be controlled via tuning parameters with clearly defined roles. 

Although the proposed model can be implemented using Markov chain Monte Carlo,  it is not scalable to high dimensional settings of our interest. As such, we implement an EM algorithm which treats the inverse covariance matrix for the shrinkage parameters as missing variables, and marginalizes over them to obtain the ``observed data" posterior which has a closed form. 
We incorporate recent computational developments such as the dynamic weighted lasso \citep{Chang2010} to obtain a computationally efficient approach which is scalable to high dimensional settings. 

We present the proposed methodology and the EM algorithm in Section 2, the theoretical results in Section 3, and the simulation results comparing our approach with several competitors in Section 4. We apply our method to a cancer genomics study in Section 5.

\section{Methodology}

\subsection{Model Specification}

Let $\bzero_m$ and $\b1_m$ denote the length-$m$ vectors with 0 entries and 1 entries, respectively, and $I_m$ the $m \times m$ identity matrix. The subscript $m$ may be omitted in the absence of ambiguity. For any length-$m$ vector $\bv$, we define $e^\bv = \left(e^{v_1},\dots,e^{v_m}\right)'$, $\log \bv = \left(\log v_1,\dots,\log v_m \right)'$, $|\bv| = \left(|v_1|,\dots,|v_m|\right)'$, and $D_{\bv} = \diag(\bv)$.

Suppose we have a random sample of $n$ observations $\{y_i, \bx_i; i=1,\ldots,n\}$ where $y_i$ is the outcome variable and $\bx_i$ is a vector of $p$ predictors. Let $\mathcal{G} = \langle V,E \rangle$ denote the known underlying graph for the $p$ predictors, where $V=\{1,\dots,p\}$ is the set of nodes and $E \subset \{(j,k): 1\le j < k \le p\}$ is the set of undirected edges. Let $G$ be the $p \times p$ adjacency matrix in which the $(j,k)$-th element $G_{jk}=1$ if there is an edge between predictors $j$ and $k$, and $G_{jk}=0$ if otherwise.

Consider the linear model
\begin{eqnarray}
\by = X \bbeta + \bepsilon, \mbox{ } \bepsilon \sim \cN(\bzero,\sigma^2I_n), \label{eq:model}
\end{eqnarray}
where $\by=(y_1,\dots,y_n)'$, $X=(\bx_1,\dots,\bx_n)'$, $\bbeta=(\beta_1,\dots,\beta_p)'$, $\bepsilon=(\epsilon_1,\dots,\epsilon_n)'$, and $\mathcal{N}(\cdot)$ denotes the Gaussian distribution. We assign the following priors to $\bbeta$ and $\sigma^2$
\begin{align}
\beta_j \sim& \mathcal{DE}(\lambda_j/\sigma),  \mbox{ } \sigma^2 \sim \mathcal{IG}(a_\sigma,b_\sigma), \mbox{ } j=1,\dots,p, \label {eq:base}
\end{align}
where $\lambda_j$ is the shrinkage parameter for $\beta_j$, $\mathcal{DE}(\cdot)$, and $\mathcal{IG}(\cdot)$ denote the double exponential (Laplace) and inverse gamma distributions, respectively.  Prior specification \eqref{eq:base} differs from Bayesian Lasso \citep{Park2008} in that the degree of shrinkage for the $j$-th coefficient is controlled by $\lambda_j$ ($j=1,\ldots,p$) not a common $\lambda$, allowing for adaptive shrinkage guided by underlying graph knowledge.

We encode the graph information $\cG$ in the model via an informative prior on the shrinkage parameters as follows.
\begin{align}
\balpha = (\log(\lambda_1), \ldots, \log(\lambda_p))' \sim \mathcal{N} \left( \bmu,\nu \mathit{\itOmega}^{-1} \right), \label{eq:shrinkage}
\end{align}
where
\begin{align*}
\mathit{\itOmega} = \left[ \begin{array}{cccc} 1+\sum_{j\neq1} \omega_{1j} & -\omega_{12} & \cdots & -\omega_{1p}\\
-\omega_{21} & 1+\sum_{j\neq2} \omega_{2j} & \ddots & -\omega_{2p}\\
\vdots & \ddots & \ddots & \vdots\\
-\omega_{p1} & -\omega_{p2} & \cdots & 1+\sum_{j\neq p} \omega_{pj} \end{array}
\right],
\end{align*}
and assign the following prior to $\bomega = \{ \omega_{jk}: j<k\}$
\begin{align} \label{prior_omega}
\pi(\bomega) \propto |\mathit{\itOmega}|^{-1/2} \prod_{G_{jk}=1} \omega_{jk}^{a_\omega-1} \exp ( -b_\omega\omega_{jk} ) 1 (\omega_{jk}>0) \prod_{G_{jk}=0} \delta_0(\omega_{jk}),
\end{align}
where $\delta_0$ is the Dirac delta function concentrated at 0 and $1(\cdot)$ is the indicator function. Since $\mathit{\itOmega}$ is symmetric and diagonally dominant, it is guaranteed to be positive definite. It follows from prior \eqref{prior_omega} that $\omega_{jk}=0$ if $G_{jk}=0$ and $\omega_{jk}>0$ if $G_{jk}=1$. In other words, under our model formulation the shrinkage parameters $\lambda_j$ and $\lambda_k$ have a positive partial correlation if predictors $j$ and $k$ are connected and have a zero partial correlation otherwise.  The magnitudes of the positive partial correlations are learned from the data, with a higher partial correlation leading to the smoothing of  corresponding shrinkage parameters. Our model formulation has several appealing features. First, a higher positive partial correlation between two connected predictors results in an increased probability of having both predictors selected or excluded simultaneously under an EM algorithm. This makes intuitive sense when both variables are important or unimportant.  Second, in the scenario where one of the connected predictors is important and the other one is not, the method can learn from the data and impose a weak partial correlation, thereby enabling the corresponding shrinkage parameters to act in a largely uncorrelated manner.  Third, the selection of unconnected variables is guided by shrinkage parameters which are partially uncorrelated. Finally, our approach does not constrain the effect sizes for connected variables to be similar in magnitude.


\begin{figure}[h!]
\includegraphics[width=\textwidth]{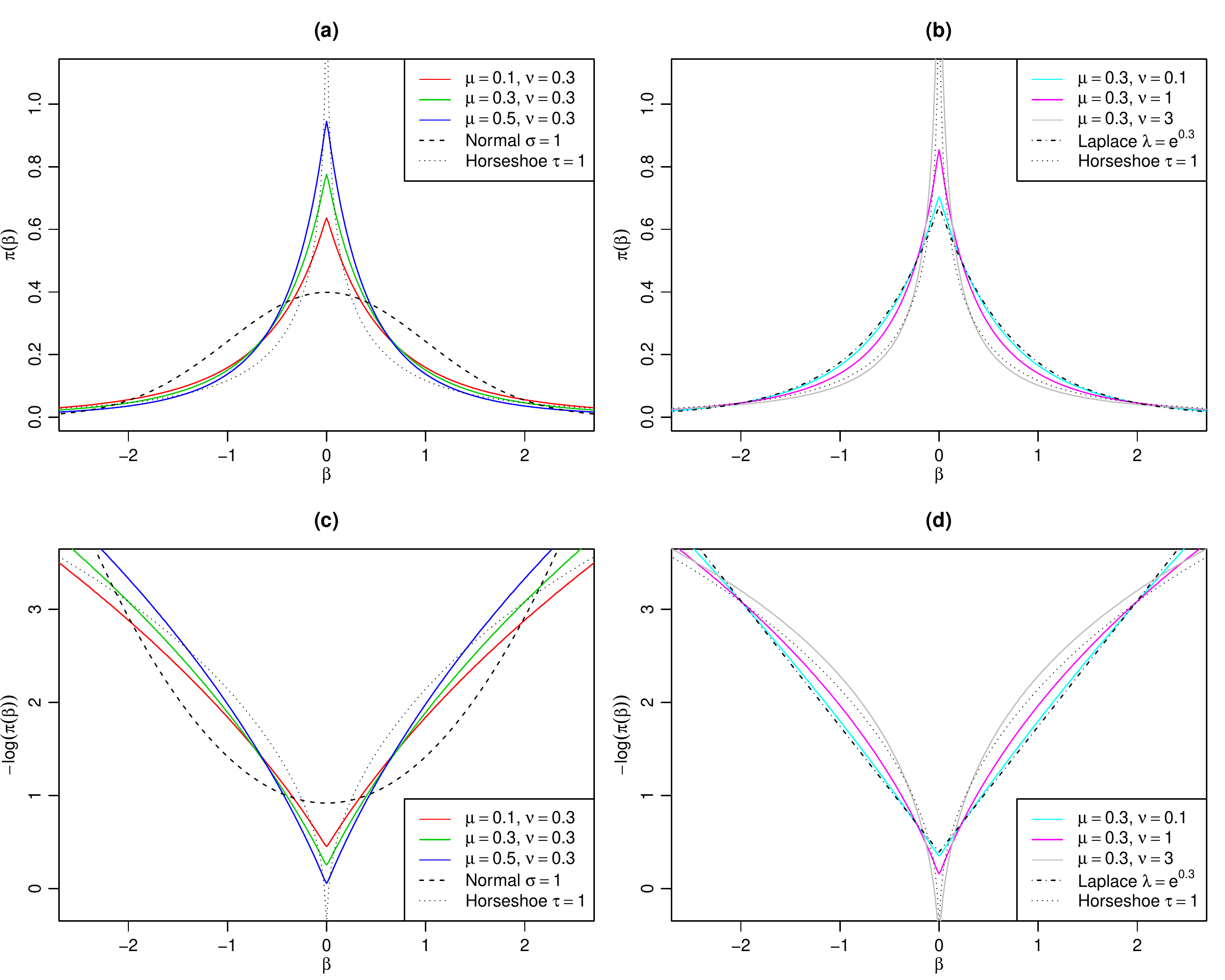}
\caption{Top two panels plot the marginal prior densities of $\beta$ for (a) different $\mu$ while $\nu$ and $\sigma$ are fixed and (b) different $\nu$ while $\mu$ and $\sigma$ are fixed. Bottom two panels (c) and (d) plot the corresponding negative log density functions. The standard normal prior and the horseshoe prior with $\tau=1$ are shown for contrast. The Laplacian prior with $\lambda=e^{0.3}$ is plotted as a comparison to the case with $\mu=0.3$ and $\nu=0.1$.}
\label{fig1}
\end{figure}


The mean vector $\bmu$ in \eqref{eq:shrinkage} determines the locations of $\balpha$, and can be interpreted as controlling the average sparsity of the model. In particular, one can choose $\bmu = \mu \b1$ for some $\mu \in \mathbb{R}$, where a greater value of $\mu$ implies a sparser model. Figure \ref{fig1}(a) plots the marginal density for the regression coefficients for different values of $\mu$ with $\lambda$ marginalized out (via Monte Carlo averaging), while $\nu$ and $\sigma$ are kept fixed. It is clear that larger $\mu$ values lead to sharper peaks at zero with lighter tails, thus encouraging greater shrinkage. On the other hand, $\nu$ specifies the prior confidence on the choice of $\bmu$ as the average sparsity parameter. If $\nu = 0$, we have $\balpha = \bmu$ so that the shrinkage parameters are fixed, resulting in a Lasso type shrinkage. This is evident from Figure \ref{fig1}(d), which plots the negative logarithm of the density for the marginal regression coefficients for different values of $\nu$ while $\mu$ and $\sigma$ are fixed. Figures 1(b) and 1(d) also show that larger values of $\nu$ result in  higher-peaked and heavier-tailed densities and the corresponding penalty becomes similar to non-convex penalties in the frequentist literature, e.g. SCAD in \citet{Fan2001}. Overall, changing the value of $\nu$ results in different types of penalty functions which can be convex or non-convex. 



We note that \eqref{prior_omega} looks similar to a product of the gamma densities. However, it involves an additional term $|\mathit{\itOmega}|^{-1/2}$ which is required to obtain a closed form full posterior, since the term cancels out between $\pi(\balpha)$ and $\pi(\bomega)$. A similar trick was used for specifying the inverse covariance matrix for the regression coefficients in \citet{Liu2014}, which they denote as a graph Laplacian structure. However our approach is distinct in that it specifies a graph Laplacian type structure for the inverse covariance matrix for the log-shrinkage parameters and incorporates prior graph knowledge. Moreover, their approach results in an OSCAR type penalty \citep{Bondell2008}, while $-\log(\pi(\beta))$ under our approach can lead to both convex and non-convex penalties depending on the value of $\nu$.

Proposition \ref{pro} shows that the prior in \eqref{prior_omega} is proper. The proof is presented in the Appendix.

\begin{pro} \label{pro}
The prior $\pi(\bomega)$ of $\bomega$ in \eqref{prior_omega} is proper.
\end{pro}

\begin{figure}[h!]
\includegraphics[width=0.5\textwidth]{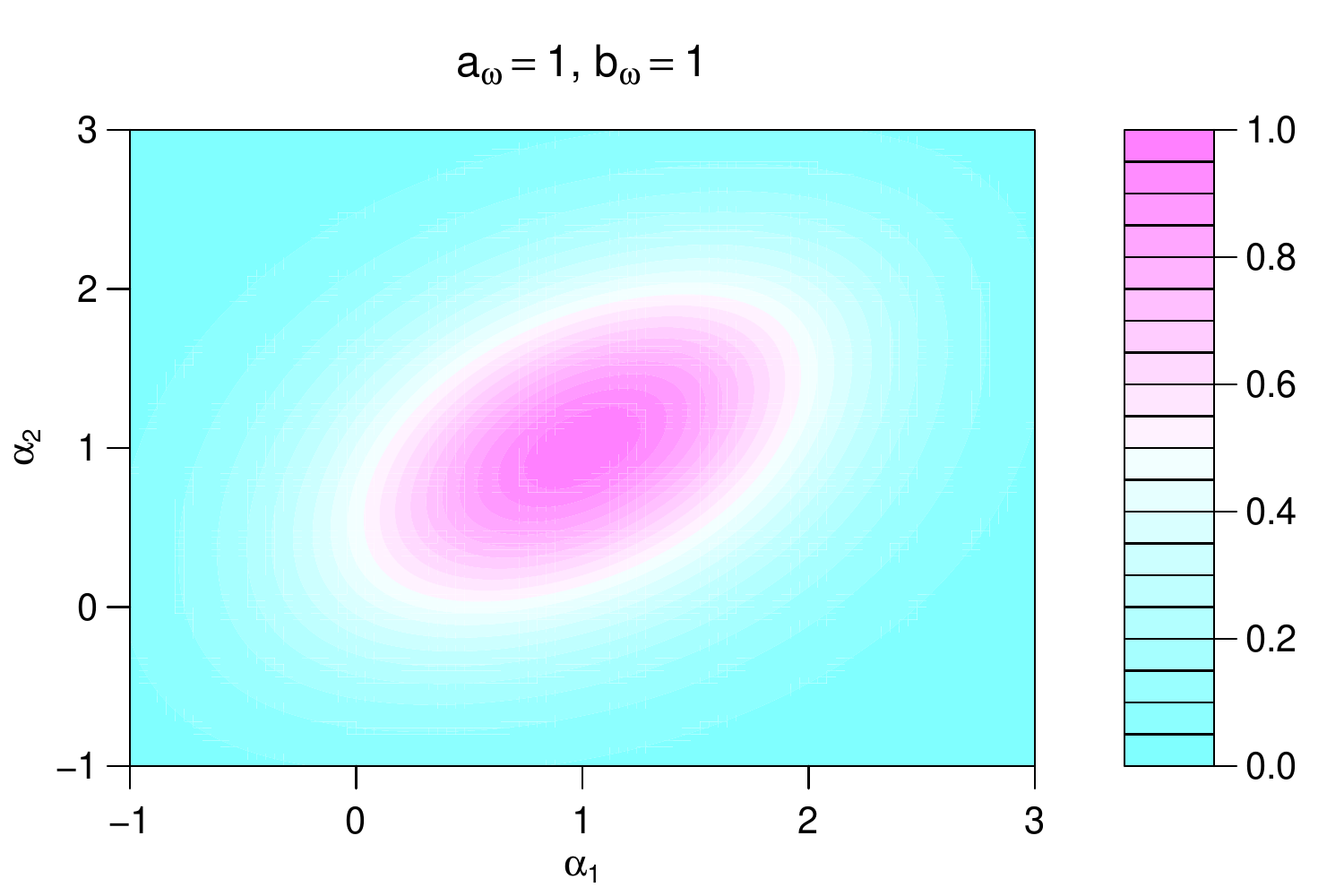}
\includegraphics[width=0.5\textwidth]{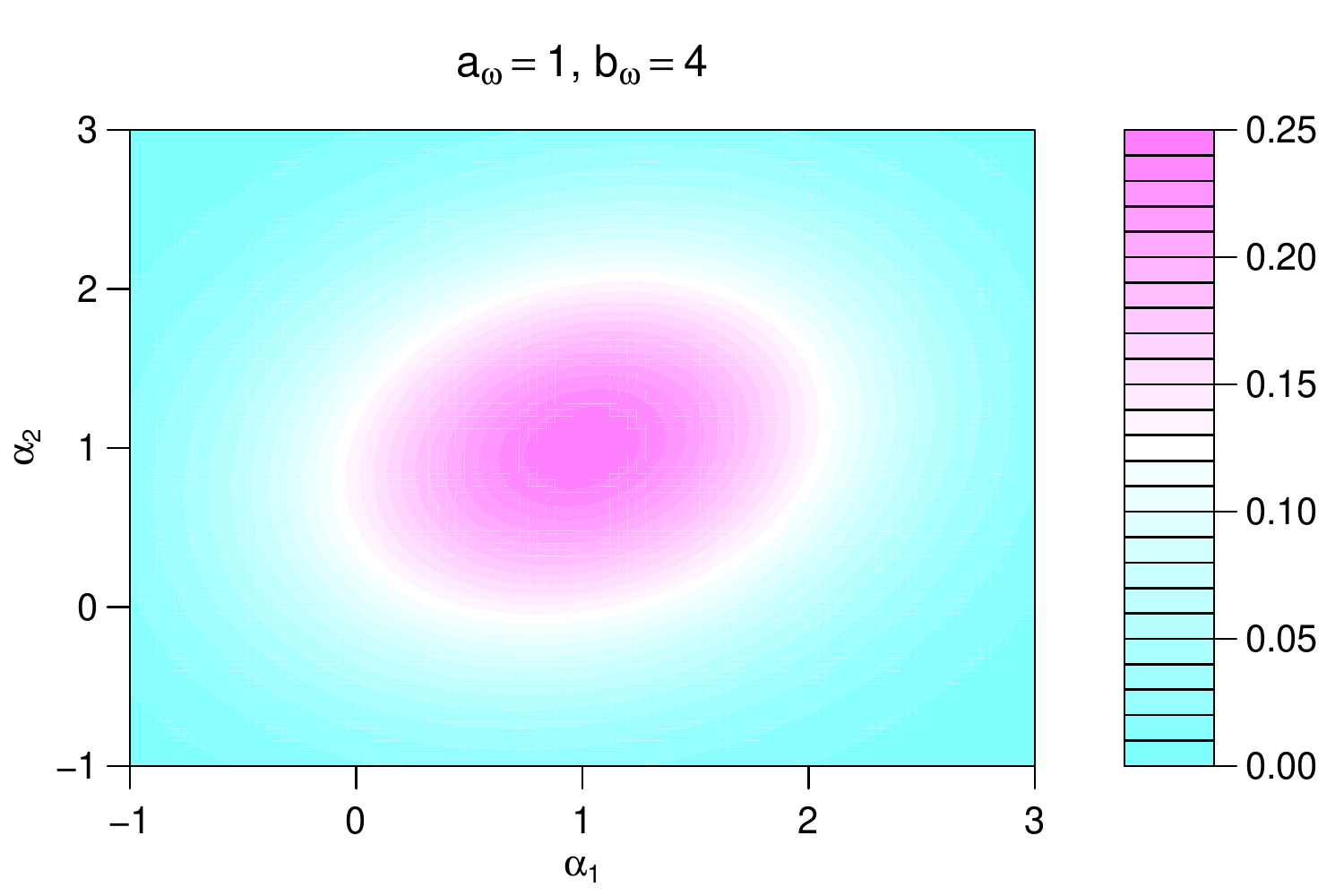}
\includegraphics[width=0.5\textwidth]{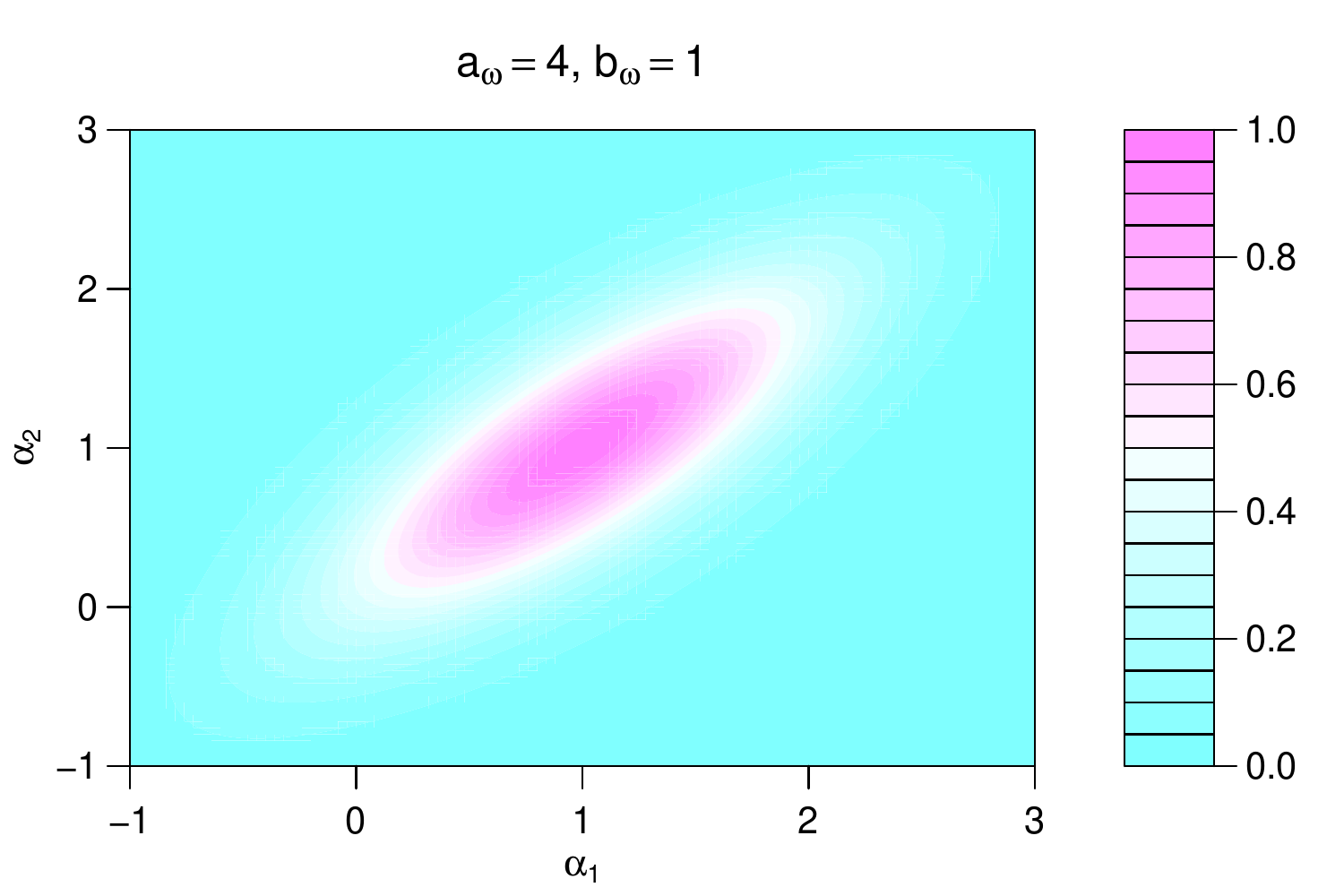}
\includegraphics[width=0.5\textwidth]{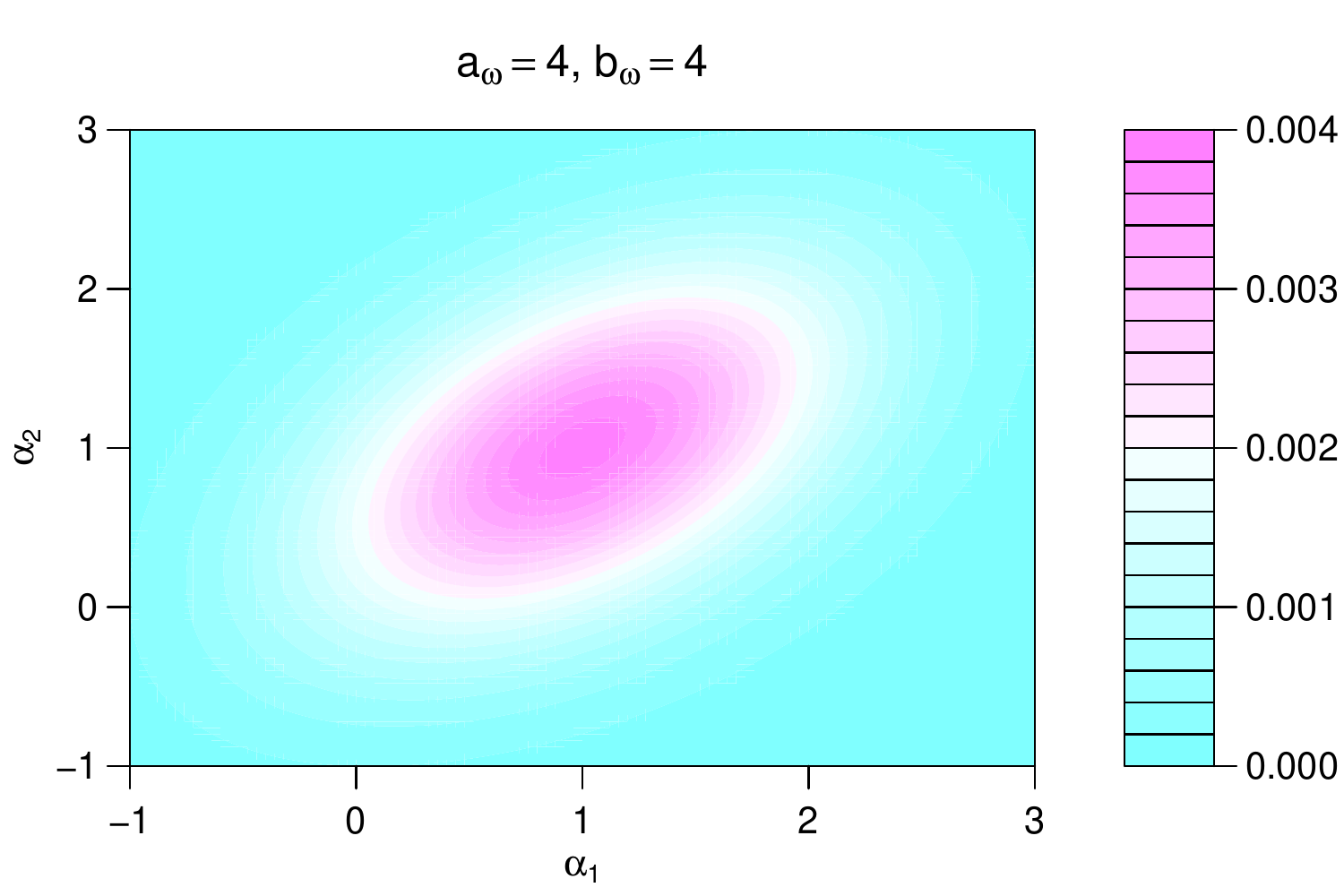}
\caption{Contour plots of the marginal prior density of $\alpha_1$ and $\alpha_2$ for 4 different combinations of $a_\omega$ and $b_\omega$.}
\label{fig2}
\end{figure}

The prior in \eqref{prior_omega} involves a shape parameter $a_\omega$ and the rate parameter $b_\omega$, which serve the similar roles as those of the gamma distribution. In fact, they are directly involved in regulating the correlations between the elements of $\balpha$. 
To see how they affect these correlations, consider $p=2$ and $G_{12}=1$. It follows that the joint prior density of $\alpha_1$ and $\alpha_2$ after marginalizing out $\omega_{12}$ is given (up to a constant) by
\begin{align*}
\pi(\alpha_1,\alpha_2) \propto f(\alpha_1,\alpha_2) = \exp \left( -\frac{(\alpha_1-\mu_1)^2+(\alpha_2-\mu_2)^2}{2\nu} \right) \left( b_\omega + \frac{(\alpha_1-\alpha_2)^2}{2\nu} \right)^{-a_\omega}.
\end{align*}
Figure \ref{fig2} draws the contour plots of $f(\alpha_1,\alpha_2)$ for 4 different combination of $a_\omega$ and $b_\omega$; $(a_\omega,b_\omega) = (1,1),(1,4),(4,1),(4,4)$ with $\mu_1=\mu_2=1$ and $\nu=1$. As $a_\omega$ increases and/or $b_\omega$ decreases, $\alpha_1$ and $\alpha_2$ tend to have a stronger correlation, translating to a higher probability of having similar values. This is also evident in the E-step in the EM algorithm (see equation \eqref{eq:Estep}), where high values of  $a_\omega/b_\omega$ tend to result in a high mean value for $\omega_{jk}$  which in turn tends to result in similar values for $\alpha_j-\mu_j$ and $\alpha_k - \mu_k$.

\subsection{EM Algorithm}

The Maximum-A-Posteriori (MAP) estimator for the proposed model is obtained by maximizing the posterior density over $\btheta = (\bbeta', \sigma^2, \balpha')'$ with $\bomega$ marginalized out. Specifically,
\begin{align}
\widehat{\btheta} = \left(\widehat{\bbeta},\widehat{\sigma}^2,\widehat{\balpha}\right) = \argmax_{\btheta} \int \pi(\btheta,\bomega|\by,X) d\bomega, \label{eq:argmax}
\end{align}
where the full posterior density is given by
\begin{align*}
\pi(\btheta,\bomega|\by,X) &\propto \pi(\by|\bbeta,\sigma^2,X) \pi(\bbeta|\sigma^2,\balpha) \pi(\sigma^2) \times |\mathit{\itOmega}|^{1/2} \exp \left( -\frac{(\balpha-\bmu)' \mathit{\itOmega} (\balpha-\bmu)}{2\nu}  \right)\\
& \qquad \times |\mathit{\itOmega}|^{-1/2} \prod_{j<k,G_{jk}=1} \omega_{jk}^{a_\omega-1} \exp ( -b_\omega\omega_{jk} ) \prod_{j<k,G_{jk}=0} \delta_0(\omega_{jk} ).
\end{align*}
In the case of $\mathit{\itOmega} = I_p$, where no graph information is used, we call the resulting estimator the \emph{EM} estimator for Bayesian \emph{SH}rinkage approach, or EMSH in short. In the general case where prior graph information is used, we call the resulting estimator the EMSH with the \emph{S}tructural information incorporated, or EMSHS in short.

We use $\bmu = \mu\b1$ where $\mu>0$ for simplicity. Note that the algorithm can be easily modified to accommodate heterogeneous sparsity parameters.

Since
\begin{align*}
(\balpha-\bmu)' \mathit{\itOmega} (\balpha-\bmu) = \sum_{j=1}^p (\alpha_j-\mu)^2 + \sum_{j<k} \omega_{jk} (\alpha_j-\alpha_k)^2,
\end{align*}
we have
\begin{align}
\pi(\btheta,\bomega|\by,X) &\propto \pi(\by|\bbeta,\sigma^2,X) \pi(\bbeta|\sigma^2,\balpha) \pi(\sigma^2) \times \exp \left( -\frac{(\balpha-\bmu)' (\balpha-\bmu)}{2\nu}  \right) \nonumber \\
& \qquad \times \prod_{j<k,G_{jk}=1} \omega_{jk}^{a_\omega-1} \exp \left( -b_\omega \omega_{jk} - \frac{\omega_{jk}}{2\nu} (\alpha_j-\alpha_k)^2 \right) \prod_{j<k,G_{jk}=0} \delta_0(\omega_{jk} ). \label{eq:joint}
\end{align}
Therefore, the marginal posterior density for $\btheta$ is given by
\begin{align}
\pi(\btheta|\by,X) & \propto \pi(\by|\bbeta,\sigma^2,X) \pi(\bbeta|\sigma^2,\balpha) \pi(\sigma^2) \times \exp \left( -\frac{(\balpha-\bmu)' (\balpha-\bmu)}{2\nu}  \right) \nonumber\\
& \qquad \times \prod_{j<k,G_{jk}=1} \left( b_\omega + \frac{1}{2\nu} (\alpha_j-\alpha_k)^2 \right)^{-a_\omega}. \label{eq:marginal}
\end{align}
Since the marginal posterior density in \eqref{eq:marginal} is differentiable with respect to $\btheta$ and the set $\{\btheta: \pi(\btheta|\by,X) \ge \eta\}$ is bounded and closed for any $\eta>0$,  its maximum is attainable and the MAP estimator always exists; see Theorem 2.28 in \citet{rudin1976principles}. Since the logarithm of marginal posterior density may not be convex, the MAP estimator may have multiple (local) solutions. However, our numerical experiments suggests a stable performance under our method, and we show in Section 3, that the algorithm admits a unique solution asymptotically.

Although one can directly optimize \eqref{eq:marginal} to compute $\widehat{\btheta}$ in \eqref{eq:argmax}, we choose to use the EM algorithm to obtain the MAP estimate. This is because the solution surface for $\balpha$ given $\bbeta$ after marginalizing out $\bomega$ in \eqref{eq:marginal} is non-convex, leading to potential computational difficulties. We elaborate more on this when describing the M-step for $\balpha$. In summary, we optimize $\pi(\btheta|\by,X)$ by proceeding iteratively with the ``complete data" log-posterior $\pi(\btheta, \bomega | \by, X)$ in \eqref{eq:joint}, where $\itOmega(\bomega)$ is considered ``missing data."  At each EM iteration, we replace $\itOmega$ by its conditional expectation in the E-step and then maximize the expected ``complete data" log posterior with respect to $\btheta$ in the M-step.

The objective function to be optimized at the $t$-th EM iteration is given by
\begin{align}
Q_t(\btheta) = & -\frac{n+p+2a_\sigma+2}{2} \log (\sigma^2) \nonumber \\
& -\frac{(\by-X\bbeta)'(\by-X\bbeta) + 2\sigma \sum_{j=1}^p e^{\alpha_j} |\beta_j| + 2b_\sigma}{2\sigma^2} \nonumber \\
& + \sum_{i=1}^p \alpha_i - \frac{(\balpha-\bmu)' \mathit{\itOmega}^{(t)} (\balpha-\bmu)}{2\nu}, \label{eq:opt}
\end{align}
where $\mathit{\itOmega}^{(t)} = \mathbb{E}\left( \mathit{\itOmega} | \by, X, \btheta^{(t-1)} \right)$.

\subsubsection{E-step} It follows from \eqref{eq:joint} that the posterior density of $\bomega$ given $\btheta$ is the product of the gamma densities where $\omega_{jk}$ follows the gamma distribution with parameters $a_\omega$ and $b_\omega+\frac{\left(\alpha_j-\alpha_k\right)^2}{2\nu}$ for $j < k, G_{jk}=1$. Therefore, we have
\begin{align}
\omega_{jk}^{(t)} = \mathbb{E}(\omega_{jk}|\by,X,\btheta^{(t-1)}) &= \frac{2\nu a_\omega G_{jk}}{ 2\nu b_\omega + \left( \alpha_j^{(t-1)} -\alpha_k^{(t-1)} \right)^2}, \qquad j<k. \label{eq:Estep}
\end{align}
Since we only need to update as many $\omega_{jk}$ as the number of edges in $\cG$, this step can be completed in $O(|E|)$ operations, which is computationally very inexpensive for sparse graphs.

\subsubsection{M-step} For this step, we sequentially optimize the objective function with respect to $\bbeta$, $\sigma^2$, and $\balpha$.

\begin{itemize}

\item M-step for $\bbeta$: With $\sigma = \sigma^{(t-1)}$ and $\balpha = \balpha^{(t-1)}$ fixed, $\bbeta^{(t)}$ can be obtained as
\begin{align*}
\bbeta^{(t)} =  \argmin_{\bbeta} \, \frac{1}{2} (\by-X\bbeta)'(\by-X\bbeta)+ \sum_{j=1}^p \xi_j |\beta_j|,
\end{align*}
where $\xi_j = \sigma e^{\alpha_j}$. This is a weighted lasso problem, which can be solved by many algorithms such as \citet{Efron2004}, \citet{wu2008coordinate}, and \citet{Chang2010}. We use the dynamic weighted lasso (DWL) algorithm developed in \citet{Chang2010}, which is capable of rapidly computing the solution by borrowing information from previous iterations when the regularization parameters change across the EM iterations. Our experience suggests that these regularization parameters differ negligibly over EM iterations under our approach, especially as the solution approaches its limit. As such, the DWL results in substantial savings in computation, compared to alternate algorithms such as LARS which needs to completely recompute the solution for each EM iteration.

Finding a lasso solution using the DWL algorithm requires $O(pq^2)$ operations where $q$ is the number of nonzero coefficients in the solution, provided that the sample correlations between the selected variables and all remaining variables are available. The latter requires an additional $O(npq)$ operations. Therefore, while the initial M-step for $\bbeta$ takes $O(npq)$ operations, the DWL algorithm updates the solution in $O(pq)$ operations as the EM iterations continue and the solution stabilizes. Readers are referred to \citet{Chang2010} for further details regarding the DWL algorithm.

We note that \citet{Park2008}, \citet{Armagan2013}, and several others used the normal mixture representation of the Laplace prior below to compute MAP estimates under an EM algorithm
\begin{align*}
\frac{\lambda}{2\sigma} e^{-\lambda|\beta|/\sigma} = \int_0^\infty \frac{1}{\sqrt{2\pi\tau\sigma^2}} e^{-\beta^2/(2\tau\sigma^2)} \frac{\lambda^2}{2} e^{-\lambda^2\tau/2} d\tau,
\end{align*}
where $\tau$ is the latent scale parameter that is imputed in the E-step. We choose to use the form of the Laplace prior instead of the above mixture representation due to several considerations. First, an M-step for $\bbeta$ of the EM algorithm under the normal mixture representation takes $O(n^2p)$ operations, which is slower than the proposed approach. Second, as pointed out by \citet{Armagan2013}, the Laplace representation leads to faster convergence than the normal mixture representation. Third, the regression coefficients cannot attain exact zeros in the normal mixture representation, and additional post-processing steps are required for variable selection, which can be sensitive to cut-off values. Lastly, numerical difficulties may arise when $\beta$ approaches zero under the normal mixture representation because the conditional mean of $\tau^{-1}$ may explode to infinity.

\item M-step for $\sigma$: With $\bbeta = \bbeta^{(t)}$ and $\balpha = \balpha^{(t-1)}$ fixed, we have
\begin{align*}
\sigma^{(t)} = \argmin_\sigma \frac{c_1}{\sigma^2} + \frac{c_2}{\sigma} + c_3 \log \sigma,
\end{align*}
where $c_1 = \frac{1}{2}(\by-X\bbeta)'(\by-X\bbeta) + b_\sigma$, $c_2 = \sum_{j=1}^p e^{\alpha_j} |\beta_j|$, and $c_3 = n+p+2a_\sigma+2$.
The solution is then given by $\sigma^{(t)} = \frac{c_2 + \sqrt{c_2^2+8c_1c_3}}{2c_3}.$

\item M-step for $\balpha$: Since there is no closed-form solution for $\balpha$, we use the Newton method. With $\bbeta = \bbeta^{(t)}$, $\sigma = \sigma^{(t)}$, and $\mathit{\itOmega} = \mathit{\itOmega}^{(t)}$ fixed, the Newton search direction at $\balpha$ is given by $\bd_N(\balpha) = -H^{-1} \bg,$ where $H = \sigma \mathit{\itOmega} + \nu D_{|\bbeta|} D_{e^{\balpha}}$ \mbox{ and } $\bg = \sigma \mathit{\itOmega} \left( \balpha-\bmu \right) - \nu \sigma \b1 + \nu D_{|\bbeta|} e^{\balpha}$. As the Hessian matrix $H$ is always positive definite, $\bd_N(\balpha)$ becomes a valid Newton direction. Therefore, we can update $\balpha$ as follows
\begin{align}\label{eq:update_alpha}
\balpha^{(t)} &= \balpha^{(t-1)} + s_t \bd_N \left( \balpha^{(t-1)} \right),
\end{align}
where $s_t$ is the step size.

Since the usual Newton method involves the inversion of the $p \times p$ Hessian matrix $H$, it is only feasible when $p$ is moderate. When $p$ is large, we suggest replacing the Hessian matrix by its diagonal matrix \citep{becker1988improving} and $\bd_N(\balpha)$ by  $\bd(\balpha) = - D_H^{-1} \bg$, where $D_H = \diag(H) = \sigma \diag(\mathit{\itOmega}) + \nu D_{|\bbeta|} D_{e^{\balpha}}$. Since $D_H$ is positive definite, $\bd(\balpha)$ is a valid descent direction, and the step size $s_t$ can be determined by the backtracking line search
\citep{nocedal2006numerical}.

Note that there are only $p+|E|$ unique nonzero elements in $\mathit{\itOmega}$. Therefore, obtaining the $p$-dimensional direction vector takes $O(p+|E|)$ operations only. Since edges in network graphs are usually sparse, its overall computation is much faster than the Newton method even if approximating the Hessian matrix may slightly increase the number of EM iterations.

In addition, it is not necessary to repeat the Newton steps until convergence to obtain the optimal solution for $\balpha$ within each M-step for $\balpha$. It suffices that each iteration of M-step for $\balpha$ ensures an increase in the value of the objective function, in order to guarantee the convergence of the EM algorithm. However, our experience indicates that repeating the Newton steps three to five times within each M-step for $\balpha$ helps reduce the number of total EM iterations.

As alluded to earlier, the advantage of the EM algorithm over directly optimizing the marginal posterior density $\pi(\btheta|\by,X)$ with respect to $\btheta$ lies in the fact that the Hessian matrix with respect to $\balpha$ is guaranteed to be positive definite in the former case, while it is not in the latter. Since the EM algorithm exploits part of the curvature information in optimizing with respect to $\balpha$ at nearly no extra computational cost, it is expected to lead to a reduced number of total iterations and hence savings in computation \citep{nocedal2006numerical}.

\end{itemize}

The EM algorithm can be started from the E-step for $\bomega$ with suggested initial values $\bbeta^{(0)} = \bzero$, $\sigma^{(0)} = \sqrt{(\by'\by+2b_\sigma)/c_3}$, and $\alpha_j^{(0)} = \mu$ for all $j$. The number of operations in each EM iteration is $O(npq+|E|)$ initially and reduces to $O(pq+|E|)$ after a few iterations. We repeat the EM procedures until the relative improvement of the optimum value of the objective function goes below a certain threshold, say $\epsilon=e^{-5}$.

\subsection{Role of Shrinkage Parameters $\balpha=\log(\blambda)$ }
\label{interpretation}


It is straightforward to show that the estimators satisfy
\begin{align}
\widehat{\bbeta} = \argmin_{\bbeta} \, \frac{1}{2}(\by-X\bbeta)'(\by-X\bbeta) + \sum_{j=1}^p \widehat{\xi}_j |\beta_j|, \label{eq:betahat}
\end{align}
and
\begin{align} \label{adaptive_penalty}
\widehat{\balpha} = \argmin_{\balpha} \, \frac{1}{2\nu} (\balpha-\bmu)' \mathit{\itOmega}^{(\infty)} (\balpha-\bmu) - \b1' \balpha + \frac{1}{\widehat{\sigma}} | \widehat{\bbeta} |' e^{\balpha},
\end{align}
where $\widehat{\xi}_j = \widehat{\sigma} e^{\widehat{\alpha}_j}$ and $\mathit{\itOmega}^{(\infty)}$ is the final value of $\mathit{\itOmega}$ from the EM algorithm. When $\widehat{\sigma}$ and $\widehat{\balpha}$ are fixed, the solution $\widehat{\bbeta}$ in \eqref{eq:betahat} resembles an adaptive lasso solution with the regularization parameter $\widehat{\bxi}$. Instead of assuming fixed weights as in the adaptive lasso, the EMSHS uses the data and the underlying graph knowledge to learn the weights. Specifically, the estimate of $\alpha_j$ depends on the shrinkage parameters corresponding to variables connected to $x_j,j=1,\ldots,p$ and the corresponding partial correlations, as follows
\begin{align}
\big|\widehat{\beta}_j\big|  = \frac{\widehat{\sigma}}{\nu} \left( \mu + \nu - \widehat{\alpha}_j + \sum_{k \sim j} \omega_{jk}^{(\infty)} (\widehat{\alpha}_k - \widehat{\alpha}_j) \right) e^{-\widehat{\alpha}_j}. \label{eq:adaptive_weights}
\end{align}
By estimating the weights in an adaptive manner guided by the prior graph knowledge, the proposed approach avoids having to specify an initial consistent estimator for the weights as in the adaptive lasso, which is expected to be of significant practical advantage in high dimensional settings. This is in fact our experience in numerical studies; see Section 4. Finally, we note that larger values of $\widehat{\alpha}_j$ translate to smaller values for $\big|\widehat{\beta}_j\big|,j=1,\ldots,p$, and vice-versa, clearly demonstrating the role of the shrinkage parameters $\balpha$.



\section{Theoretical Properties} \label{sec:oracle}



To fix ideas, let $p_n$ denote the number of candidate predictors,  of which $q_n$ are the true important variables. Model~\eqref{eq:model} is reformulated as
\begin{align*}
\by_n = X_n \bbeta_0 + \bepsilon_n,
\end{align*}
where $\by_n$ is the $n \times 1$ response vector, $X_n$ is the $n\times p_n$ design matrix, $\bbeta_0$ is the $p_n \times 1$ true coefficient vector, and $\bepsilon_n$ is the $n \times 1$ error vector. The errors are independent Gaussian with mean 0 and variance $\sigma_0^2$; $\bepsilon_n \sim \cN(\bzero,\sigma_0^2I_n)$, and the errors are also independent of the covariates.

The covariates are stochastic and are dictated by an inverse covariance matrix depending on a true graph $\cG_{0n}$. They are standardized such that
\begin{align*}
\b1' \bx_{nj} = 0, \qquad \bx_{nj}'\bx_{nj} = n, \qquad j=1,\dots,p_n,
\end{align*}
where $\bx_{nj}$ is the $j$-th column (variable) of $X_n$, and let $\itSigma_n = \frac{1}{n} X_n' X_n$ be the sample covariance matrix.

Let $\widehat{\btheta}_n = (\widehat{\bbeta}_n',\widehat{\sigma}_n^2,\widehat{\balpha}_n')'$ be the EMSHS solution.
Let $\cA_n = \{j:\widehat{\beta}_{nj} \neq 0\}$ be the index set of the selected variables in $\widehat{\bbeta}_n$,
and $\cA_0 = \{j:\beta_{0j} \neq 0\}$ be the index set of the true important variables where $|\cA_0|=q_n$.
We assume $\|\bbeta_0\|$ is bounded so that the variance of the response and the signal-to-noise ratio stay bounded. Without loss of generality, we assume $\|\bbeta_0\| = 1$.

For any index set $\cA$, $\bv_\cA$ represents the subvector of a vector $\bv$ with entries corresponding to $\cA$. $E_{\cA\cB}$ is the submatrix of a matrix $E$ with rows and columns corresponding to $\cA$ and $\cB$, respectively. When a sequential index set $\cA_n$ is used for a sequence of vectors or matrices indexed by $n$, the subscript $n$ may be omitted for conciseness if it does not cause a confusion. For example, $\bv_{n\cA_n}$ can be written as $\bv_{\cA_n}$ or $\bv_{n\cA}$, and $E_{n\cA_n\cB_n}$ can be written as $E_{\cA_n\cB_n}$ or $E_{n\cA\cB}$.

Let $O( \cdot )$, $o( \cdot )$, $O_p( \cdot )$, and $o_p( \cdot )$ denote the standard big $O$, little $o$, big $O$ in probability, and little $o$ in probability, respectively.
Further $f(n) = \Theta(g(n))$ indicates that $f(n)$ and $g(n)$ satisfy $f(n) = O(g(n))$ and $g(n) = O(f(n))$; $f(n) = \Theta_p(g(n))$ is similarly defined.
When these notations are used for vectors and matrices, they bound the $L_2$-norm $\| \cdot \|$ of the entities. For example, $\bv = O(n)$ means that $\|\bv\| = O(n)$. Every norm $\| \cdot \|$ in this article denotes the $L_2$ norm. Finally, $\rightarrow_p$ and $\rightarrow_d$ denote convergence in probability and in distribution, respectively.

\subsection{Oracle Property for Fixed $p$}

Consider the case with a fixed number of candidate predictors (i.e., $p_n=p$). Suppose the following conditions hold as $n \rightarrow \infty$.

\begin{enumerate}[label=(A\arabic{enumi}),ref=(A\arabic{enumi})]
\item \label{ass:fix:beta} $\|\bbeta_0\| = 1$ and $\min_{j \in \cA_0} |\beta_{0j}| \ge C_\beta$ for some constant $C_\beta>0$.
\item \label{ass:fix:XX} $\itSigma_n \rightarrow_p \itSigma_0$ where $\itSigma_0$ is positive definite and depends on $\cG_{0n}=\cG_0$.
\item \label{ass:fix:mu} $\mu_n = R\log n + o(\log n)$ where $1/2<R<1$.
\item \label{ass:fix:nu} $\nu_n = \Theta( n^{-r} \log n )$ where $0 < r < R-1/2$.
\item \label{ass:fix:omega} $a_{\omega n} b_{\omega n}^{-1} = o(1)$.
\item \label{ass:fix:sigma} $a_{\sigma n} = a_{\sigma 1} n^z$ and $b_{\sigma n} = b_{\sigma 1} n^z$ for $a_{\sigma 1}>0$, $b_{\sigma 1} > 0$, and $0 \le z < 1$.
\end{enumerate}

Assumption \ref{ass:fix:beta} states that the nonzero coefficients stay away from zero, although their magnitudes are allowed to vary with $n$. Assumption \ref{ass:fix:XX} is a fairly general regularity condition for the
design matrix which rules out collinearity between covariates, and ensures that the important variables are not replaced by any other remaining variables in the model.  Assumption \ref{ass:fix:XX} is a fairly general regularity condition for the design matrix. Readers are referred to Remark \ref{rem:fix:principle} for the comments on \ref{ass:fix:mu} and \ref{ass:fix:nu}. Assumption \ref{ass:fix:omega} forces the precision matrix to assume a diagonal form as $n \to \infty$.
Thus as $n \to \infty$, we essentially do not need to utilize prior graph knowledge $\cG_0$ to establish the theoretical results for fixed $p$ case. Hence our asymptotic results for fixed $p$ is agnostic to the structure of the prior graph, and hence robust to mis-specification. However, we note that for finite samples, incorporation of true prior graph knowledge is of paramount importance in achieving improved numerical performance.  According to \ref{ass:fix:sigma}, the prior on $\sigma^2$ is well-tightened, with $\sigma^2$ converging to a constant when $z>0$.

\begin{thm} \label{oracle:fix}
Assume the conditions \ref{ass:fix:beta}-\ref{ass:fix:sigma}.
The following statements hold for the EMSHS estimator $\widehat{\btheta}_n = (\widehat{\bbeta}_n',\widehat{\sigma}_n^2,\widehat{\balpha}_n')'$ as $n \rightarrow \infty$.
\begin{enumerate}[label=(\alph*)]
\item $P ( \cA_n = \cA_0 ) \rightarrow 1$.
\item $n^{1/2} \left( \widehat{\bbeta}_{n\cA_0} - \bbeta_{0\cA_0} \right) \rightarrow_d \mathcal{N}\left(\bzero, \sigma_0^2 \itSigma_{0\cA_0\cA_0}^{-1} \right)$.
\item The solution is unique in probability.
\end{enumerate}
\end{thm}


The proof for Theorem 1 is provided in Appendix.

\begin{rem}
Although we only consider Gaussian errors, the results also hold for error distributions with finite variance.
\end{rem}

\begin{rem} \label{rem:fix:principle}
With a little lack of rigorousness, \ref{ass:fix:mu} ensures $\widehat{\xi}_{nj} = \widehat{\sigma}_n e^{\widehat{\alpha}_{nj}} = \Theta_p(n^R)$ for $j \in \cA_n^c$, and \ref{ass:fix:nu} ensures $|\widehat{\beta}_{nj}| \widehat{\xi}_{nj} = \Theta_p(n^r)$ for $j \in \cA_n$.
Therefore, if $\widehat{\bbeta}_n$ is $\sqrt{n}$-consistent, which is indeed the case, the important variables receive shrinkage of order $\widehat{\xi}_{nj} = \Theta_p(n^r)$ and the unimportant variables receive shrinkage of order \emph{at least} $\widehat{\xi}_{nj} = \Theta_p(n^{r+1/2})$. This is the key that leads to the oracle property.
\end{rem}

\begin{rem} \label{rem:fix:sigma}
The true residual variance $\sigma_0^2$ is consistently estimated by $\widehat{\sigma}_n^2$. That is, $\widehat{\sigma}_n^2 \rightarrow_p \sigma_0^2$.
\end{rem}

\subsection{Oracle Property for Diverging $p$}

When the number of candidate predictors is diverging, let $\cG_n = \langle V_n, E_n \rangle$ be the working graph which is used to fit the model where $V_n = \{1,\dots,p_n\}$ and $E_n$ is the set of edges. Let $G_n$ be the adjacency matrix for $\cG_n$, $l_{nj} = \sum_{k=1}^{p_n} G_{n,jk}$ be the degree of the vertex $j$ in $\cG_n$, and $L_n = \max_{1 \le j \le p_n} l_{nj}$ be the maximum degree among all vertices.

Suppose the following conditions hold as $n \rightarrow \infty$.

\begin{enumerate}[label=(B\arabic{enumi}),ref=(B\arabic{enumi})]
\item \label{ass:p} $p_n = O(\exp(n^U))$ where $0 \le U < 1$.
\item \label{ass:q} $q_n = O(n^u)$ where $0 \le u < (1-U)/2$ and $q_n \le p_n$.
\item \label{ass:beta} $\|\bbeta_0\| = 1$ and $\min_{j\in \cA_0} |\beta_{0j}| \ge C_\beta q_n^{-1/2}$ for some constant $C_\beta>0$.
\item \label{ass:XX} Assume that $\cG_{0n}$ is such that the smallest eigen value of $\itSigma_{n\cA\cA}$ is greater than $\tau_1$ for any index set $\cA$ with $|\cA| \le n$, and that the largest eigen value of $\itSigma_n$ is less than $\tau_2$ almost surely, where $0<\tau_1<\tau_2<\infty$.
\item \label{ass:rho} Assume that $\cG_{0n}$ is such that the following partial orthogonality condition holds.
\begin{align*}
\|\itSigma_{n\cB\cC}\|^2 \le \rho_n^2 \|\itSigma_{n\cB\cB}\| \|\itSigma_{n\cC\cC}\|, \qquad \forall \cB \subset \cA_0, \forall \cC \subset \cA_0^c,
\end{align*}
almost surely where $\rho_n = O(n^{-1/2})$.
\item \label{ass:mu} $\mu_n = R\log n + \frac{1}{2} \log (1+p_n/n) + o(\log n)$ where $(U+1)/2<R<1-u$.
\item \label{ass:nu} $\nu_n = \Theta( (1+p_n/n)^{-1} n^{-r} \log n )$ where $0 < r < R-1/2< 1/2-u$.
\item \label{ass:omega} $L_n a_{\omega n} b_{\omega n}^{-1} = o(1)$.
\item \label{ass:sigma} $a_{\sigma n} = a_{\sigma 1} n^z$ and $b_{\sigma n} = b_{\sigma 1} n^z$ for $a_{\sigma 1}>0$, $b_{\sigma 1} > 0$, and $1-r<z<1$.
\end{enumerate}

Assumption~\ref{ass:p} allows the number of candidate predictors to increase at an exponential rate and \ref{ass:q} allows the number of important variables to diverge as well.  Assumption \ref{ass:beta} states that the $L_2$ norm of the true regression coefficients is bounded, which in conjunction with diverging $q_n$ implies that some of the true nonzero coefficients may get sufficiently small. However, \ref{ass:beta} ensures that they remain  away from zero sufficiently. In order to accommodate increasing $p_n$ and $q_n$, the shrinkage parameters need to be carefully calibrated, which is ensured under conditions \ref{ass:mu} and \ref{ass:nu} on $\mu_n$ and $\nu_n$ and by the fact that $q_n$ increases at a moderate rate in \ref{ass:q}. The roles of $\mu_n$ and $\nu_n$ are further explained in Remark \ref{rem:rate}.

\ref{ass:XX} is analogous to \ref{ass:fix:XX} for the fixed $p$ case. The partial orthogonality condition in \ref{ass:rho} assumes that the unimportant variables are asymptotically weakly correlated with the important variables; similar assumptions are widely used for the case of diverging p in the literature \citep{Huang2008}. Since $p_n\to\infty$, the degree of a vertex in the graph $\mathcal{G}_n$ can diverge. In order to precisely regulate the smoothing effect between neighboring shrinkage parameters, the condition \ref{ass:fix:omega} needs to be extended to \ref{ass:omega} which incorporates information about the degrees of vertices in the working graph $\cG_n$. The condition \ref{ass:sigma} is stronger than \ref{ass:fix:sigma} in order to prevent $\widehat{\sigma}^2_n$ from converging to zero much faster than desired: see Remark \ref{rem:sigma} for further comments on $\widehat{\sigma}^2_n$.

\begin{thm} \label{thm:oracle}
Assume the conditions \ref{ass:p}-\ref{ass:sigma}.
The following statements hold for the EMSHS estimator $\widehat{\btheta}_n = (\widehat{\bbeta}_n',\widehat{\sigma}_n^2,\widehat{\balpha}_n')'$ as $n \rightarrow \infty$.
\begin{enumerate}
\item $P( \cA_n = \cA_0 ) \rightarrow 1$.
\item Letting $s_n^2 = \bgamma_n' \itSigma_{n\cA_0\cA_0}^{-1} \bgamma_n$ for any sequence of $q_n \times 1$ nonzero vectors $\bgamma_n$, we have
\begin{align*}
n^{1/2} s_n^{-1} \bgamma_n' ( \widehat{\bbeta}_{n\cA_0} - \bbeta_{0\cA_0} ) \rightarrow_d \mathcal{N} (0,\sigma_0^2).
\end{align*}
\item The solution is unique in probability.
\end{enumerate}
\end{thm}

We note that, in contrast to the fixed $p$ case, the oracle property result for the diverging $p$ case requires assumptions on the true graph, as in conditions \ref{ass:XX} and \ref{ass:rho}, as well as knowledge about the working graph $\cG_n$ in \ref{ass:omega}. The proof is provided in the Appendix.

\begin{rem}
Although we only consider Gaussian errors, the results can be readily generalized to moderately heavier tailed errors.
\end{rem}

\begin{rem} \label{rem:rate}
In parallel with Remark \ref{rem:fix:principle}, with a little lack of rigorousness, \ref{ass:mu} ensures $\widehat{\xi}_{nj} = \Theta_p(n^R)$ for $j \in \cA_n^c$, and \ref{ass:nu} ensures $|\widehat{\beta}_{nj}| \widehat{\xi}_{nj} = \Theta_p(n^r)$ for $j \in \cA_n$. If $\widehat{\bbeta}_n$ is $\sqrt{n}$-consistent, which is indeed the case, the important variables receive shrinkage of order \emph{at most} $\widehat{\xi}_{nj} = \Theta_p(n^r q_n^{1/2})$ due to \ref{ass:beta} and the unimportant variables receive shrinkage of order \emph{at least} $\widehat{\xi}_{nj} = \Theta_p(n^{r+1/2})$.
\end{rem}

\begin{rem} \label{rem:sigma}
Unlike Remark \ref{rem:fix:sigma}, $\widehat{\sigma}_n^2$ may converge to 0. However, once rescaled, $\widehat{\sigma}_n^2$ consistently estimates the true residual variance $\sigma_0^2$; $(n+p_n)\widehat{\sigma}_n^2/n \rightarrow_p \sigma_0^2$.
\end{rem}

\section{Simulation Study} \label{simulation}
We conduct simulations to evaluate the performance of the proposed approach in comparison with several existing methods. The competing methods include the lasso (Lasso), the adaptive Lasso (ALasso) \citep{Zou2006}, the Bayesian variable selection approach using spike and slab priors and MRF priors by \citet{Stingo2011} which we denote as BVS-MRF, and finally the EM approach for Bayesian variable selection (denoted as EMVS) proposed by \citet{Rockova2014} and its extension to incorporate structural information (denoted as EMVSS). Of note, EMSHS, EMVSS and BVS-MRF incorporate the graph information, whereas the other methods do not. For Lasso and ALasso, we use the glmnet R package where the initial consistent estimator for ALasso is given by the ridge regression. The Matlab code for the MCMC approach is provided with the original article by \citet{Stingo2011}. \citet{Rockova2014} provided us their unpublished R codes for EMVS and EMVSS.

\subsection{Simulation Set-up}

The simulated data are generated from the following model
\begin{align*}
y_i = \bx_i' \bbeta + \epsilon_i, \qquad 1 \le i \le n,
\end{align*}
where $\bx_i \sim \cN(\bzero,\mathit{\itSigma}_X)$, $\epsilon_i \sim \cN(0,\sigma_\epsilon^2)$, and $\bbeta = (\underbrace{1,\dots,1}_{q},\underbrace{0,\dots,0}_{p-q})$.
The first $q=5$ variables are the important variables and the last $p-q$ variables are unimportant variables. The sample size is fixed at $n=50$; the residual variance is fixed at $\sigma_\epsilon=1$; and we consider $p=1{,}000$, $10{,}000$, and $100{,}000$.

Let $G_0$ be the adjacency matrix for the true covariate graph, which determines ${\itSigma}_X$. That is, $G_{0,jk} = 1$ if there is an edge between predictors $j$ and $k$, and $G_{0,jk} = 0$ otherwise. $G_0$ is generated as follows.
\begin{enumerate}[label=(\arabic*)]
\item We generate $g$ virtual pathways, depending on the total number of predictors; $g=50$ for $p=1{,}000$, $g=300$ for $p=10{,}000$ and $p=100{,}000$.
\item The first pathway is composed of the $q$ important variables only.
\item The number of genes in other pathways are negative binomial random variables with mean $\mu_{path} = 30$.
\item The genes which belong to a pathway are chosen randomly and independently of other pathways. Hence the pathways can be overlapping.
\item Edges are randomly generated ensuring all genes in a pathway have at least one path to all the other genes in the pathway. This can be done by conducting the following procedure for each pathway.
\begin{enumerate}[label=(\alph*)]
\item Randomly choose two genes and insert an edge between the two. Mark the two genes as connected. Mark the others as unconnected.
\item \label{stepX} Randomly choose a connected gene and an unconnected gene, and add an edge between them. Mark the unconnected gene as connected.
\item Repeat step \ref{stepX} until all genes are connected. This will form a tree, where all genes have at least one path to all the other genes in the pathway.
\item In order to add some extra edges, for each pair of genes that do not share an edge, add an edge between them with probability $p_1 = 0.05$. $p_1$ determines the overall density of edges.
\end{enumerate}
\end{enumerate}

Given $G_0$, the covariance matrix $\mathit{\itSigma}_X$ is designed as follows.
 \begin{enumerate}[label=(\roman*)]
\item Set $A = I_p$.
\item Calculate the vertex degrees $D_j = \sum_{k=1}^p G_{0,jk}$.
\item \label{step3} For each pair $j<k$ with $G_{0,jk} = 1$, set $A_{jk} = A_{kj} = -S_{jk}/(\max(D_j,D_k) \times 1.1 + 0.1)$ where
\begin{align*}
S_{jk} = \begin{cases}
1, & \textrm{if } 1 \le j,k \le q,\\
\mathrm{Ber}(1/2), & \textrm{otherwise}.
\end{cases}
\end{align*}
\item Set $\mathit{\itSigma}_X = A^{-1}$ and then rescale $\mathit{\itSigma}_X$ so that its diagonal elements become 1.
\end{enumerate}
Note that the resulting covariance matrix $\mathit{\itSigma}_X$ is diagonally dominant and positive definite and $X_j$ and $X_k$ are partially correlated only if $G_{0,jk} = 1$.
Also note that since this procedure involves inverting a $p \times p$ matrix, we used this method for $p=1{,}000$ and $p=10{,}000$ cases only. For $p=100{,}000$ case, the network information of the first $10,000$ variables were generated by this procedure, and the second set of independent $90{,}000$ variables were added and they were independent of the first set of variables.

Let $G$ be the adjacency matrix of the pathway graph that is used to fit the model. We consider several scenarios where the graph used to fit the model may be specified correctly or mis-specified, as follows
\begin{enumerate}[label=\arabic*)]
\item $G_0$ is as described above and $G=G_0$.
\item $G_0$ is as described above but allows no edge between important variables and unimportant variables, and $G=G_0$.
\item $G_0$ is the same as in scenario (1), but $G$ is randomly generated with the same number of edges as $G_0$.
\item $G_0$ is the same as in scenario (2), but $G$ is randomly generated with the same number of edges as $G_0$.
\item $G_0$ is the same as in scenario (1), but $G$ includes only a subset of the edges in $G_0$ for which the corresponding partial correlations are greater than $0.5$.
\end{enumerate}
Scenarios (1) and (2) are cases where the true graph is completely known; scenario (2) allows no correlation between important and unimportant variables and hence is an ideal setting for our approach. In scenarios (3) and (4) considered as the worst case scenario, $G$ is completely mis-specified. Scenario (5) mimics the situation where only strong signals from $G_0$ are known to data analysts, which is between the ideal and the worst case scenarios.

For the proposed approach, we choose an uninformative prior $\sigma^2 \sim \cIG(1,1)$. Based on our numerical studies which show that the performance of EMSH and EMSHS is not highly sensitive to $a_\omega$, $b_\omega$, and $\nu$, we recommend using $a_\omega=4$, $b_\omega=1$, and $\nu=1.2$.

We generate 500 simulated datasets in total, each of which contains a training set, validation set, and test set of size $n=50$ each. We fit the model using the training data for a grid of values of the tuning parameter $\mu$ lying between $(3.5,7.5)$, and then we choose the value that minimizes the prediction error for the validation data. Variable selection performance is assessed in terms of the number of false positives (FP) and the number of false negatives (FN) and prediction performance is assessed in terms of mean squared prediction error (MSPE) calculated using the test data. We also report the average computation time per tuning parameter value in seconds.

\subsection{Results}
The simulation results are summarized in Tables \ref{tbl1}, \ref{tbl2}, and \ref{tbl3}. BVS-MRF is omitted in the case of $p=10,000$ and both BVS-MRF and EMVSS are omitted in $p=100{,}000$ cases, because they are not scalable or errors are reported when applied to these settings.

All methods achieve better performance in scenario 2 (or 4) compared to scenario 1  (or 3), indicating that the problem is more challenging in the presence of nonzero correlation between the important variables and the unimportant variables.  Within each of scenarios 1, 2, and 5, where true or partially true graph knowledge is available, the structured variable selection methods EMVSS and EMSHS are superior to their counterparts that do not use graph information (i.e. EMVS and EMSH). Moreover, the performance of each structured variable selection method (namely, EMSHS, EMVSS and BVS-MRF) improves from Scenario 3 (or 4) to Scenario 1 (or 2), further demonstrating the benefits of the correctly specified graph information. Similarly, the partially correctly specified graph information in Scenario 5 also improves the performance of these methods compared to Scenario 3. In Scenarios 3 and 4,

For prediction, the EMSHS yields the best performance in all settings when the graph information is correctly or partially correctly specified (i.e., scenarios 1, 2 and 5).  When the graph information is completely mis-specified (scenarios 3 and 4), EMSHS still yields the best or close to the best performance in all settings, demonstrating its robustness to mis-specified graph information. In addition, EMSH performs the best among the unstructured variable selection methods for all cases under $p=1{,}000$ and $p= 10{,}000$, lending support to the advantage of using the data to learn adaptive shrinkage in EMSH as discussed in section \ref{interpretation}.

For variable selection, the EMSHS yields the best or close to the best performance in all settings when the graph information is correctly or partially correctly specified. Of note, while EMVS tends to have close to 0 false positives, it has high false negatives. In addition, the false positives under the proposed methods are significantly lower compared to Lasso and adaptive lasso. Finally, EMSHS consistently yields lower false positives and false negatives than EMSH in scenarios 1, 2 and 5, which again demonstrates the advantage of incorporating true graph information.

While the difference in performance between EMSHS and EMVSS in terms of prediction and variable selection is subtle for $p=1{,}000$ in scenarios 1 and 5, they behave somewhat differently for $p=10{,}000$. EMSHS tends to have relatively lower false negatives admitting slightly higher false positives. Including more important variables seems to have led to smaller prediction errors than EMVSS.

Although somewhat slower than the Lasso and adaptive Lasso, the proposed structured variable selection approach is still computationally efficient and is scalable to $p=100{,}000$ and higher dimensions, which is substantially better than the BVS-MRF and EMVSS.

\begin{small}
\begin{table}
\caption{\label{tbl1} The mean squared prediction error (MSPE) for the test data, the number of false positives (FP), the number of false negatives (FN), and the average computation time per tuning parameter value in seconds are recorded for $p=1{,}000$ case. In the parentheses are the corresponding standard errors.}

\begin{tabular}{l|cccr}
Method & MSPE & FP & FN & \multicolumn{1}{c}{Time}\\
\hline
& \multicolumn{4}{l}{Scenario \#1: Reference case}\\
Lasso & 2.29 (0.04) & 24.11 (0.48) & 0.04 (0.01) & 0.00\\
ALasso & 2.12 (0.04) & 12.20 (0.35) & 0.14 (0.02) & 0.00\\
BVS-MRF & 3.41 (0.07) & 10.56 (0.72) & 1.23 (0.05) & 540.78\\
EMVS & 5.58 (0.11) & 0.00 (0.00) & 3.55 (0.06) & 0.50\\
EMVSS & 1.36 (0.02) & 1.39 (0.07) & 0.05 (0.01) & 1.30\\
EMSH & 1.76 (0.04) & 2.62 (0.14) & 0.27 (0.03) & 0.07\\
EMSHS & 1.31 (0.03) & 1.13 (0.09) & 0.06 (0.02) & 0.75\\
\hline
& \multicolumn{4}{l}{Scenario \#2: Ideal case}\\
Lasso & 1.73 (0.02) & 18.11 (0.48) & 0.00 (0.00) & 0.00\\
ALasso & 1.48 (0.02) & 6.17 (0.22) & 0.01 (0.00) & 0.00\\
BVS-MRF & 2.02 (0.04) & 7.18 (0.58) & 0.28 (0.03) & 539.82\\
EMVS & 3.11 (0.10) & 0.00 (0.00) & 1.93 (0.06) & 0.47\\
EMVSS & 1.22 (0.01) & 0.76 (0.05) & 0.01 (0.00) & 1.25\\
EMSH & 1.28 (0.02) & 0.71 (0.06) & 0.09 (0.02) & 0.06\\
EMSHS & 1.14 (0.01) & 0.24 (0.05) & 0.00 (0.00) & 0.74\\
\hline
& \multicolumn{4}{l}{Scenario \#3: Worst case (Reference)}\\
Lasso & 2.21 (0.04) & 23.61 (0.47) & 0.03 (0.01) & 0.00\\
ALasso & 2.04 (0.04) & 11.50 (0.34) & 0.12 (0.02) & 0.00\\
BVS-MRF & 3.39 (0.07) & 10.39 (0.73) & 1.23 (0.06) & 516.85\\
EMVS & 5.41 (0.11) & 0.01 (0.00) & 3.49 (0.06) & 0.51\\
EMVSS & 1.83 (0.05) & 2.62 (0.11) & 0.30 (0.03) & 1.23\\
EMSH & 1.66 (0.04) & 2.74 (0.14) & 0.19 (0.02) & 0.06\\
EMSHS & 1.73 (0.04) & 5.41 (0.31) & 0.22 (0.03) & 0.68\\
\hline
& \multicolumn{4}{l}{Scenario \#4: Worst case (Ideal)}\\
Lasso & 1.72 (0.03) & 18.42 (0.48) & 0.01 (0.00) & 0.00\\
ALasso & 1.48 (0.02) & 6.04 (0.22) & 0.03 (0.01) & 0.00\\
BVS-MRF & 1.99 (0.04) & 7.35 (0.59) & 0.28 (0.03) & 408.52\\
EMVS & 3.13 (0.10) & 0.00 (0.00) & 1.91 (0.06) & 0.48\\
EMVSS & 1.35 (0.02) & 1.17 (0.07) & 0.09 (0.01) & 1.20\\
EMSH & 1.29 (0.02) & 0.76 (0.06) & 0.10 (0.02) & 0.06\\
EMSHS & 1.51 (0.03) & 4.32 (0.28) & 0.19 (0.02) & 0.69\\
\hline
& \multicolumn{4}{l}{Scenario \#5: Intermediate case}\\
Lasso & 2.25 (0.04) & 22.91 (0.47) & 0.05 (0.01) & 0.00\\
ALasso & 2.06 (0.04) & 11.42 (0.31) & 0.13 (0.02) & 0.00\\
BVS-MRF & 3.36 (0.07) & 12.28 (0.79) & 1.21 (0.06) & 449.23\\
EMVS & 5.41 (0.11) & 0.01 (0.00) & 3.45 (0.06) & 0.47\\
EMVSS & 1.34 (0.03) & 1.27 (0.07) & 0.04 (0.01) & 1.96\\
EMSH & 1.66 (0.04) & 2.55 (0.15) & 0.21 (0.02) & 0.06\\
EMSHS & 1.31 (0.03) & 1.33 (0.12) & 0.05 (0.01) & 0.81\\
\hline
\hline
\end{tabular}
\end{table}
\end{small}

\begin{small}
\begin{table}
\caption{\label{tbl2}  The mean squared prediction error (MSPE) for the test data, the number of false positives (FP), the number of false negatives (FN), and the average computation time per tuning parameter value in seconds are recorded for $p=10{,}000$ case. In the parentheses are the corresponding standard errors.}
\begin{tabular}{l|cccr}
Method & MSPE & FP & FN & \multicolumn{1}{c}{Time}\\
\hline
& \multicolumn{4}{l}{Scenario \#1: Reference case}\\
Lasso & 3.34 (0.07) & 29.00 (0.53) & 0.34 (0.03) & 0.03\\
ALasso & 3.21 (0.08) & 15.87 (0.46) & 0.54 (0.04) & 0.02\\
EMVS & 6.88 (0.12) & 0.00 (0.00) & 4.17 (0.04) & 33.70\\
EMVSS & 2.01 (0.07) & 2.23 (0.13) & 0.56 (0.04) & 83.52\\
EMSH & 2.98 (0.09) & 4.74 (0.26) & 1.04 (0.05) & 0.96\\
EMSHS & 1.94 (0.08) & 2.71 (0.25) & 0.39 (0.04) & 5.56\\
\hline
& \multicolumn{4}{l}{Scenario \#2: Ideal case}\\
Lasso & 2.30 (0.04) & 23.68 (0.55) & 0.04 (0.01) & 0.03\\
ALasso & 2.05 (0.05) & 9.16 (0.30) & 0.16 (0.02) & 0.02\\
EMVS & 5.36 (0.16) & 0.00 (0.00) & 3.15 (0.06) & 36.64\\
EMVSS & 1.43 (0.03) & 1.02 (0.07) & 0.21 (0.02) & 80.68\\
EMSH & 1.79 (0.05) & 2.10 (0.15) & 0.43 (0.04) & 0.80\\
EMSHS & 1.16 (0.02) & 0.64 (0.11) & 0.01 (0.01) & 4.67\\
\hline
& \multicolumn{4}{l}{Scenario \#3: Worst case (Reference)}\\
Lasso & 3.31 (0.07) & 28.10 (0.54) & 0.34 (0.03) & 0.03\\
ALasso & 3.17 (0.07) & 15.65 (0.47) & 0.57 (0.04) & 0.03\\
EMVS & 7.13 (0.13) & 0.00 (0.00) & 4.23 (0.04) & 28.85\\
EMVSS & 3.27 (0.08) & 3.97 (0.15) & 1.37 (0.05) & 56.38\\
EMSH & 2.91 (0.08) & 4.91 (0.26) & 1.03 (0.05) & 0.87\\
EMSHS & 3.04 (0.08) & 9.34 (0.45) & 1.04 (0.05) & 4.34\\
\hline
& \multicolumn{4}{l}{Scenario \#4: Worst case (Ideal)}\\
Lasso & 2.24 (0.04) & 24.25 (0.55) & 0.04 (0.01) & 0.03\\
ALasso & 1.98 (0.04) & 8.45 (0.32) & 0.14 (0.02) & 0.03\\
EMVS & 5.14 (0.15) & 0.00 (0.00) & 3.01 (0.06) & 31.66\\
EMVSS & 1.84 (0.05) & 2.22 (0.10) & 0.48 (0.03) & 54.44\\
EMSH & 1.72 (0.04) & 2.13 (0.15) & 0.35 (0.03) & 0.72\\
EMSHS & 1.94 (0.05) & 7.55 (0.39) & 0.40 (0.03) & 4.59\\
\hline
& \multicolumn{4}{l}{Scenario \#5: Intermediate case}\\
Lasso & 3.26 (0.07) & 27.68 (0.54) & 0.32 (0.03) & 0.03\\
ALasso & 3.11 (0.07) & 14.75 (0.46) & 0.54 (0.04) & 0.03\\
EMVS & 6.95 (0.12) & 0.00 (0.00) & 4.17 (0.04) & 24.49\\
EMVSS & 1.94 (0.06) & 2.27 (0.11) & 0.46 (0.03) & 63.24\\
EMSH & 2.81 (0.07) & 4.66 (0.25) & 0.95 (0.05) & 0.74\\
EMSHS & 1.72 (0.06) & 2.47 (0.22) & 0.26 (0.03) & 5.37\\
\hline
\end{tabular}
\end{table}
\end{small}

\begin{small}
\begin{table}
\caption{\label{tbl3}  The mean squared prediction error (MSPE) for the test data, the number of false positives (FP), the number of false negatives (FN), and the average computation time per tuning parameter value in seconds are recorded for $p=100{,}000$ case. In the parentheses are the corresponding standard errors.}
\begin{tabular}{l|cccr}
Method & MSPE & FP & FN & \multicolumn{1}{c}{Time}\\
\hline
& \multicolumn{4}{l}{Scenario \#1: Reference case}\\
Lasso & 4.87 (0.09) & 30.31 (0.65) & 1.21 (0.05) & 0.12\\
ALasso & 4.77 (0.09) & 16.91 (0.57) & 1.56 (0.06) & 0.18\\
EMSH & 4.66 (0.11) & 4.83 (0.29) & 2.35 (0.06) & 8.66\\
EMSHS & 3.28 (0.12) & 3.67 (0.28) & 1.26 (0.07) & 18.57\\
\hline
& \multicolumn{4}{l}{Scenario \#2: Ideal case}\\
Lasso & 3.23 (0.07) & 28.67 (0.61) & 0.29 (0.03) & 0.12\\
ALasso & 3.07 (0.08) & 11.68 (0.37) & 0.53 (0.04) & 0.13\\
EMSH & 2.93 (0.10) & 2.85 (0.19) & 1.26 (0.05) & 7.38\\
EMSHS & 1.43 (0.07) & 1.02 (0.12) & 0.10 (0.03) & 16.55\\
\hline
& \multicolumn{4}{l}{Scenario \#3: Worst case (Reference)}\\
Lasso & 5.08 (0.09) & 29.82 (0.65) & 1.31 (0.06) & 0.11\\
ALasso & 4.99 (0.10) & 17.20 (0.61) & 1.69 (0.06) & 0.16\\
EMSH & 4.89 (0.10) & 5.14 (0.32) & 2.47 (0.06) & 8.26\\
EMSHS & 4.94 (0.10) & 6.12 (0.37) & 2.49 (0.06) & 13.82\\
\hline
& \multicolumn{4}{l}{Scenario \#4: Worst case (Ideal)}\\
Lasso & 3.34 (0.07) & 28.09 (0.60) & 0.27 (0.03) & 0.11\\
ALasso & 3.18 (0.08) & 11.71 (0.39) & 0.57 (0.04) & 0.15\\
EMSH & 3.01 (0.09) & 2.92 (0.19) & 1.23 (0.06) & 7.01\\
EMSHS & 3.07 (0.09) & 4.99 (0.30) & 1.23 (0.06) & 12.46\\
\hline
& \multicolumn{4}{l}{Scenario \#5: Intermediate case}\\
Lasso & 5.07 (0.09) & 29.67 (0.63) & 1.29 (0.06) & 0.11\\
ALasso & 4.99 (0.10) & 16.36 (0.57) & 1.66 (0.06) & 0.16\\
EMSH & 4.83 (0.11) & 4.87 (0.30) & 2.41 (0.07) & 8.11\\
EMSHS & 3.55 (0.13) & 3.85 (0.30) & 1.44 (0.08) & 15.62\\
\hline
\end{tabular}
\end{table}
\end{small}

\section{Data Application}

We applied the proposed method to analysis of a glioblastoma data set obtained from the The Cancer Genome Atlas Network \citep{verhaak2010integrated}. The data set includes survival times ($T$) and gene expression data for $p=12{,}999$ genes ($X$) for 303 glioblastoma patients. As glioblastoma is known as one of the most aggressive cancers,  only $12\%$ of the samples were censored. We removed the censored observations, resulting in a sample size of $n=267$ for analysis. We fit an accelerated failure time (AFT) model as follows
\begin{align*}
\log T_i = \beta_1 X_{i1} + \cdots + \beta_p X_{ip} + \epsilon_i, \qquad i=1,\dots,n,
\end{align*}
where $\epsilon_i$'s are independent Gaussian random variables and all variables were standardized to have mean 0 and variance 1.  The network information ($\mathcal{G}$) for $X$ was retrieved from the Kyoto Encyclopedia of Genes and Genomes (KEGG) database including a total of 332 KEGG pathways and $31{,}700$ edges in these pathways.

In addition to EMSHS and EMSH, we included several competing methods that are computationally feasible, namely, lasso, adaptive lasso, EMVS, and EMVSS. The optimal tuning parameters were chosen by minimizing the 5-fold cross-validated mean-squared prediction error. The tuning parameter $\mu$ had 20 candidate values ranging from 5.5 to 6.5 ensuring solutions with various sparsity to be considered. We used $a_\sigma=1$ and $b_\sigma=1$ for prior of $\sigma^2$, which is uninformative. As shown in Table \ref{tbl4}, EMSHS achieves the best prediction performance followed by EMSH and both are substantially less expensive than EMVS and EMVSS in terms of computation. Similar to our simulation results, EMSH again yields better prediction performance than adaptive lasso, demonstrating the advantage of using the data to learn adaptive shrinkage in EMSH.

\begin{small}
\begin{table}
\caption{\label{tbl4} Cross-validated mean squared prediction error (CV MSPE) and computation time in seconds per tuning parameter (Time) from the analysis of TCGA genomic data and KEGG pathway information.}
\begin{tabular}{l|cr}
Method & CV MSPE & \multicolumn{1}{c}{Time}\\
\hline
Lasso & 0.986 & 0.2\\
ALasso & 0.996 & 0.4\\
EMVS & 0.996 & 1346.6\\
EMVSS & 0.982 & 8284.1\\
EMSH & 0.979 & 14.3\\
EMSHS & 0.975 & 17.0\\
\hline
\end{tabular}
\end{table}
\end{small}

To assess variable selection performance of EMSHS and EMSH, we conducted a second analysis. We randomly divided the entire sample into two subsets; the first subset with $187$ subjects (70\% of the whole sample) was used as the training data to fit the model and the second subset with 30\% of the subjects was used as the validation data to select optimal tuning parameter values. We repeated this procedure 100 times, resulting in 100 EMSHS and 100 EMSH solutions. Of the 100 random splits, 28 genes were selected at least once by EMSH and 21 genes were selected at least once by EMSHS. Further examination reveals that the set of genes that were selected by EMSH but not by EMSHS belong to pathways in which most of the genes were not selected. This lends support to the notion that incorporating graph information may reduce false positives, which is consistent with the findings in our simulations where EMSHS tends to yield lower false positives than
EMSH in simulation scenarios 1, 2 and 5.

The 3 genes most frequently selected by EMSHS are TOM1L1, RANBP17, and BRD7. In this set of genes the Wnt signaling pathway \citep{Kandasamy2010} was identified as an enriched pathway by the ToppGene Suite \citep{Chen2009}.  Abnormalities in the Wnt signaling pathway have been associated with human malignancies in the literature. For example, BRD7 has been shown to be correlated with enlarged lateral ventricles in mice and highly expressed in gliomas \citep{Tang2014}. TOM1L1 depletion has been shown to decrease tumor growth in xenografted nude mice \citep{Sirvent2012}. EMSHS reported average estimated coefficients of $-0.146$, $0.181$, and $0.193$ for TOM1L1, RANBP17, and BRD7, respectively. The signs of the coefficients are consistent with the known knowledge about the genes in promoting/suppressing the development of cancer. Our data analyses demonstrate that EMSHS yields biologically meaningful results.


\section{Discussion}

This article introduces a scalable Bayesian regularized regression approach and an associated EM algorithm which can incorporate the structural information between covariates in high dimensional settings. The approach relies on specifying informative priors on the log-shrinkage parameters of the Laplace priors on the regression coefficients, which results in adaptive regularization. The method does not rely on initial estimates for weights as in adaptive lasso approaches, which provides computational advantages in higher dimensions as demonstrated in simulations. Appealing theoretical properties for both fixed and diverging dimensions are established, under very general assumptions, even when the true graph is mis-specified. The method demonstrates encouraging numerical performance in terms of scalability, prediction and variable selection, with significant gains when the prior graph is correctly specified, and a robust performance under prior graph mis-specification.  

Extending the current approach to more general types of outcomes such as binary or categorical should be possible \citep{mccullagh1989generalized}, although the complexity of the optimization problem may increase. These issues can potentially be addressed using a variety of recent advances in literature involving EM approaches via latent variables \citep{polson2013bayesian}, coordinate descent method \citep{wu2008coordinate}, and other optimization algorithms \citep{nocedal2006numerical} which are readily available to facilitate computation. We leave this task as a future research question of interest. 


\section*{Acknowledgements}

This research is partly supported by NIH/NCI grants (R03CA173770 and R03CA183006). The content is solely the responsibility of the authors and does not necessarily represent the official views of the NIH. The authors thank Prof. David Dunson for helpful discussions and comments that led to the motivation of this work and Dr. Veronika Rockova for providing their code for EMVS and EMVSS.

\bibliographystyle{rss}
\bibliography{EMSHSRefs}{}

\section*{Appendix}

\subsection*{Proof of Proposition \ref{pro}}
\begin{proof}[Proof of Proposition \ref{pro}]
Note that $\mathit{\itOmega} = I + W D_{\bomega} W'$ where $W$ is a $p \times p(p-1)/2$ matrix, whose column $\bw_{jk}$ corresponding to the edge $(j,k)$ is $\be_j-\be_k$.
Here $\be_j$ is the $p \times 1$ coordinate vector whose $j$-th element is 1 and all others are zero.
Therefore, we have $|\mathit{\itOmega}| \ge 1$ and
\begin{align*}
\int |\mathit{\itOmega}|^{-1/2} \prod_{G_{jk}=1} \omega_{jk}^{a_\omega-1} \exp ( -b_\omega \omega_{jk} ) \b1 (\omega_{jk}>0) \prod_{G_{jk}=0} \delta_0(\omega_{jk}) d\bomega \le \Gamma(a_\omega)^{|E|} b_\omega^{-a_\omega|E|} < \infty.
\end{align*}
\end{proof}

\subsection*{Proof of Theorem \ref{thm:oracle}}

We first prove Theorem \ref{thm:oracle} as it is more general than Theorem \ref{oracle:fix}. We then prove Theorem \ref{oracle:fix} as a special case. Note from the M-step for $\bbeta$ that
\begin{align} \label{sol:lasso}
\widehat{\bbeta}_n =  \argmin_{\bbeta} \, \frac{1}{2} (\by_n - X_n \bbeta)'(\by_n-X_n\bbeta) + \sum_{j=1}^{p_n} \widehat{\xi}_{nj} |\beta_j|,
\end{align}
where $\widehat{\bxi}_n = \widehat{\sigma}_n \widehat{\blambda}_n = \widehat{\sigma}_n e^{\widehat{\balpha}_n}$. Then, by the Karush-–Kuhn–-Tucker (KKT) conditions (see, for example, \citet{Chang2010}), the solution is given by
\begin{align}
\label{sol:beta1} \widehat{\bbeta}_{\cA_n} &= (X_{\cA_n} ' X_{\cA_n})^{-1} ( X_{\cA_n}' \by - S_{n\cA\cA} \widehat{\bxi}_{\cA_n} ),\\
\label{sol:beta0} \widehat{\bbeta}_{\cA_n^c} &= \bzero,
\end{align}
and satisfies
\begin{align}
\label{sol:inactive} | \bx_{nj}' ( \by_n - X_n \widehat{\bbeta}_n ) | \le \widehat{\xi}_{nj}, \qquad j \notin \cA_n,
\end{align}
where
\begin{align} \label{sol:sign}
S_n =\diag(\sign(\widehat{\beta}_{n1}),\dots,\sign(\widehat{\beta}_{np_n}))
\end{align}
is the sign matrix of $\widehat{\bbeta}_n$. The M-step for $\sigma$ yields
\begin{align} \label{sol:sigma}
\widehat{\sigma}_n = \frac{\widehat{c}_{2n} + \sqrt{\widehat{c}_{2n}^2+8\widehat{c}_{1n}c_{3n}}}{2c_{3n}},
\end{align}
where $\widehat{c}_{1n} = \frac{1}{2}(\by_n - X_n \widehat{\bbeta}_n)'(\by_n - X_n \widehat{\bbeta}_n) + b_{\sigma n}$, $\widehat{c}_{2n} = \sum_{j=1}^{p_n} e^{\widehat{\alpha}_{nj}} |\widehat{\beta}_{nj}|$, and $c_{3n} = n+p_n+2a_{\sigma n}+2$.
In addition, from the M-step for $\balpha$, the solution satisfies
\begin{align} \label{sol:alpha}
|\widehat{\beta}_{nj}| e^{\widehat{\alpha}_{nj}} = \frac{\widehat{\sigma}_n}{\nu_n} \left( \mu_n+\nu_n - \widehat{\alpha}_{nj} + \sum_{k \sim j} \omega_{jk}^{(\infty)} (\widehat{\alpha}_{nk} - \widehat{\alpha}_{nj}) \right), \qquad j=1,\dots,p_n.
\end{align}

Let $\cB_n = \cA_n \cap \cA_0$, $\cC_n = \cA_n - \cA_0$, and $\cD_n = \cA_0 - \cA_n$.
The residual vector and the SSE are given by
\begin{align} \label{solution:residual:full}
\by_n - X_n \widehat{\bbeta}_n = ( I - H_{\cA_n} ) ( X_{\cD_n} \bbeta_{0\cD_n} + \bepsilon_n ) + X_{\cA_n} \left( X_{\cA_n}' X_{\cA_n} \right)^{-1} S_{n\cA\cA} \widehat{\bxi}_{\cA_n},
\end{align}
and
\begin{align}
\nonumber \| \by_n - X_n \widehat{\bbeta}_n \|^2 &= ( X_{\cD_n} \bbeta_{0\cD_n} + \bepsilon_n ) '( I - H_{\cA_n} ) ( X_{\cD_n} \bbeta_{0\cD_n} + \bepsilon_n )\\
\label{solution:SSE:full} & \qquad  + \widehat{\bxi}_{\cA_n}' S_{n\cA\cA} \left( X_{\cA_n}' X_{\cA_n} \right)^{-1} S_{n\cA\cA} \widehat{\bxi}_{\cA_n},
\end{align}
where $H_{\cA_n} = X_{\cA_n} \left( X_{\cA_n}' X_{\cA_n} \right)^{-1} X_{\cA_n}'$.

Due to the partial orthogonality \ref{ass:rho}, the solution $\widehat{\bbeta}_n$ can be rewritten as
\begin{align} \label{solution:beta}
\left[ \begin{array}{c} \widehat{\bbeta}_{\cB_n}\\ \widehat{\bbeta}_{\cC_n} \end{array} \right] = \left[ \begin{array}{c} \bbeta_{0\cB_n} + (X_{\cB_n}'X_{\cB_n})^{-1} (X_{\cB_n}' (X_{\cD_n} \bbeta_{0\cD_n} + \bepsilon_n) - S_{n\cB\cB} \widehat{\bxi}_{\cB_n})\\
(X_{\cC_n}'X_{\cC_n})^{-1} (X_{\cC_n}'\bepsilon_n  - S_{n\cC\cC} \widehat{\bxi}_{\cC_n}) \end{array} \right] + O_p(\rho_n)
\end{align}
where the last term $O_p(\rho_n)$ exists only when $\cC_n \neq \emptyset$.
The residual vector if $\cC_n = \emptyset$ is given by
\begin{align} \label{solution:residual}
\by_n - X_n \widehat{\bbeta}_n = ( I - H_{\cB_n} ) ( X_{\cD_n} \bbeta_{0\cD_n} + \bepsilon_n ) + X_{\cB_n} \left( X_{\cB_n}' X_{\cB_n} \right)^{-1} S_{n\cB\cB} \widehat{\bxi}_{\cB_n}.
\end{align}


\begin{lem} \label{lemma:sigma}
The following statements are true.
\begin{enumerate}
\item $\widehat{\sigma}_n^2 = O_p((1+p_n/n)^{-1})$.
\item $\widehat{\sigma}_n^{-2} = O_p((1+p_n/n)n^{1-z})$.
\end{enumerate}
\end{lem}
\begin{proof}
\begin{enumerate}
\item Since $\widehat{\bbeta}_n$ is the solution of \eqref{sol:lasso} and since $\bbeta = \bzero$ is a possible solution, we have
\begin{align*}
\widehat{c}_{1n} + \widehat{\sigma}_n \widehat{c}_{2n} = \frac{1}{2} \|\by_n - X_n \widehat{\bbeta}_n\|^2 + b_{\sigma n} + \widehat{\sigma}_n \sum_{j=1}^{p_n} e^{\widehat{\alpha}_{nj}} |\widehat{\beta}_{nj}| \le \frac{1}{2} \|\by_n\|^2 + b_{\sigma n}.
\end{align*}
Since
\begin{align*}
\|\by_n\|^2 = \|X_n \bbeta_0\|^2 + 2\bbeta_0'X_n'\bepsilon_n + \|\bepsilon_n\|^2  = \tau_2 n + \sigma_0^2n + O_p(n^{1/2}) = O_p(n),
\end{align*}
we have $\widehat{c}_{1n} = O_p(n)$ and $\widehat{\sigma}_n \widehat{c}_{2n} = O_p(n)$.
Note that
\begin{align} \label{sigma_bound}
\sqrt{\frac{2\widehat{c}_{1n}}{c_{3n}}} \le \widehat{\sigma}_n = \frac{\widehat{c}_{2n} + \sqrt{\widehat{c}_{2n}^2 + 8\widehat{c}_{1n}c_{3n}} }{2c_{3n}} \le \frac{\widehat{c}_{2n}}{c_{3n}} + \sqrt{\frac{2\widehat{c}_{1n}}{c_{3n}}}.
\end{align}
Since, if $0 \le b \le x \le a+b$, it follows that
\begin{align*}
x^2 \le ax+bx \le ax+b(a+b) \le 2ax + b^2,
\end{align*}
we have
\begin{align*}
\widehat{\sigma}_n^2 \le \frac{2 \widehat{\sigma}_n \widehat{c}_{2n}}{c_{3n}} + \frac{2\widehat{c}_{1n}}{c_{3n}} = O_p((1+p_n/n)^{-1}).
\end{align*}
\item Since $\widehat{c}_{1n} \ge b_{\sigma n}$, the result follows by the lower bound in \eqref{sigma_bound} and \ref{ass:sigma}.
\end{enumerate}
\end{proof}


\begin{lem} \label{lemma:alpha}
The following statements are true.
\begin{enumerate}
\item $\|\widehat{\bbeta}_n\| = O_p(1)$.

\item $\max_{1 \le j \le p_n} \widehat{\alpha}_{nj} \le \mu_n+\nu_n$.

\item $\min_{1 \le j \le p_n} \widehat{\alpha}_{nj} \ge \frac{1}{2} \log (1+p_n/n) + (r-1/2) \log n + o_p(\log n)$.
\end{enumerate}
\end{lem}
\begin{proof}
\begin{enumerate}
\item As $\widehat{c}_{1n} = O_p(n)$ in Lemma \ref{lemma:sigma}, note that
\begin{align*}
\|\by_n-X_n\widehat{\bbeta}_n\|^2 = \|\bepsilon_n\|^2 - 2(\widehat{\bbeta}_n-\bbeta_0)'X_n'\bepsilon_n + \|X_n(\widehat{\bbeta}_n-\bbeta_0)\|^2 = O_p(n).
\end{align*}
This implies $\|X_n(\widehat{\bbeta}_n-\bbeta_0)\|^2 = O_p(n)$. Since $\|\bbeta_0\|=1$ and
\begin{align*}
\|X_{\cA_n} \widehat{\bbeta}_{\cA_n}\| \le \|X_n(\widehat{\bbeta}_n-\bbeta_0)\| + \|X_n \bbeta_0\|,
\end{align*}
the result follows by \ref{ass:XX}.

\item Let $j_1 = \argmax_j \widehat{\alpha}_{nj}$. Due to \eqref{sol:alpha}, note that
\begin{align*}
\widehat{\alpha}_{nj_1}-\mu_n-\nu_n \le \sum_{k\sim j_1} \omega_{j_1k}^{(\infty)} ( \widehat{\alpha}_{nk} - \widehat{\alpha}_{nj_1} ) \le 0.
\end{align*}
Therefore, we have $\widehat{\alpha}_{nj_1} \le \mu_n + \nu_n$.

\item Let $j_0 = \argmin_j \widehat{\alpha}_{nj}$. Due to \eqref{sol:alpha}, note that
\begin{align*}
|\widehat{\beta}_{nj_0}| e^{\widehat{\alpha}_{nj_0}} \ge \widehat{\sigma}_n \frac{\mu_n + \nu_n - \widehat{\alpha}_{nj_0}}{\nu_n}.
\end{align*}
By Lemma \ref{lemma:sigma}(b),  Lemma \ref{lemma:alpha}(a), and \ref{ass:nu}, note that
\begin{align*}
\widehat{\alpha}_{nj_0} &\ge -\log |\widehat{\beta}_{nj_0}| + \log \widehat{\sigma}_n + \log (\mu_n+\nu_n-\widehat{\alpha}_{nj_0}) - \log \nu_n\\
&\ge (r-1/2) \log n + \frac{1}{2} \log (1+p_n/n) + \log (\mu_n+\nu_n-\widehat{\alpha}_{nj_0}) + o_p(\log n).
\end{align*}
By \ref{ass:mu}, we have
\begin{align*}
\mu_n - \widehat{\alpha}_{nj_0} + \log (\mu_n+\nu_n-\widehat{\alpha}_{nj_0}) \le (R-r+1/2)\log n + o(\log n).
\end{align*}
Note that $\nu_n = o(\log n)$ by \ref{ass:nu}. Since $R-r+1/2>0$, we have
\begin{align*}
\mu_n - \widehat{\alpha}_{nj_0} \le (R-r+1/2)\log n + o(\log n).
\end{align*}
Hence, the result follows.
\end{enumerate}
\end{proof}


\begin{lem} \label{lemma:xi}
The following statements are true.
\begin{enumerate}
\item $\max_{1 \le j \le p_n} \widehat{\xi}_{nj} = o_p(\widehat{\sigma}_n (1+p_n/n)^{1/2} n^{1-u})$.

\item $M_n = \max_{1 \le j \le p_n} |m_{nj}| = o(\log n)$ where $m_{nj} = \sum_{k\sim j} \omega_{jk}^{(\infty)} ( \widehat{\alpha}_{nk} - \widehat{\alpha}_{nj} )$.

\item If $|\widehat{\beta}_{nj}| = 0$ for large $n$, then we have
\begin{align*}
\widehat{\xi}_{nj} > C_2 \widehat{\sigma}_n (1+p_n/n)^{1/2} n^{R-\zeta},
\end{align*}
for $\forall C_2>0$, $\forall \zeta>0$, and large $n$.

\item If $ |\widehat{\beta}_{nj}| \le C_1 n^{-c}$ for $\exists C_1>0$, $\forall c<R-r$, and large $n$, we have
\begin{align*}
\widehat{\xi}_{nj} > C_2 \widehat{\sigma}_n^2 (1+p_n/n) n^{r+c-\zeta},
\end{align*}
for $\forall C_2>0$, $\forall \zeta>0$, and large $n$, and we have
\begin{align*}
|\widehat{\beta}_{nj}| \widehat{\xi}_{nj} \le C_3 \widehat{\sigma}_n^2 (1+p_n/n) n^r,
\end{align*}
for $\exists C_3>0$ and large $n$.

\item If $ |\widehat{\beta}_{nj}| \ge C_1 n^{-c}$ for $\exists C_1>0$, $\forall c<R-r$, and large $n$, we have
\begin{align*}
\widehat{\xi}_{nj} < C_2 \widehat{\sigma}_n^2 (1+p_n/n) n^{r+c+\zeta},
\end{align*}
for $\forall C_2>0$, $\forall \zeta>0$, and large $n$, and we have
\begin{align*}
|\widehat{\beta}_{nj}| \widehat{\xi}_{nj} \ge C_3 \widehat{\sigma}_n^2 (1+p_n/n) n^r,
\end{align*}
for $\exists C_3>0$ and large $n$.

\end{enumerate}
\end{lem}

\begin{proof}
\begin{enumerate}
\item By Lemma \ref{lemma:alpha}(b), the claim follows due to \ref{ass:mu} and \ref{ass:nu}.

\item Note that, by Lemma \ref{lemma:alpha}, we have
\begin{align*}
|\widehat{\alpha}_{nk} - \widehat{\alpha}_{nj}| \le \max_j \widehat{\alpha}_{nj} - \min_j \widehat{\alpha}_{nj} = O_p(\log n).
\end{align*}
On the other hand, we have $\omega_{jk}^{(\infty)} \le a_{\omega n} b_{\omega n}^{-1}$ by \eqref{eq:Estep}. Then, by \ref{ass:omega}, we have
\begin{align*}
M_n \le \max_{1 \le j \le p_n} \sum_{k \sim j} \omega_{jk}^{(\infty)} | \widehat{\alpha}_{nk} - \widehat{\alpha}_{nj} | \le L_n a_{\omega n} b_{\omega n}^{-1} O(\log n) = o(\log n).
\end{align*}

\item If $|\widehat{\beta}_{nj}| = 0$, then we have $\widehat{\alpha}_{nj} = \mu_n + \nu_n + m_{nj}$. The claim follows by \ref{ass:mu}, \ref{ass:nu}, and part (b).

\item By \eqref{sol:alpha} and \ref{ass:nu}, note that
\begin{align*}
\widehat{\alpha}_{nj} &= -\log |\widehat{\beta}_{nj}| + \log \widehat{\sigma}_n + \log (\mu_n+\nu_n-\widehat{\alpha}_{nj} + m_{nj}) - \log \nu_n\\
&\ge (r+c)\log n + \log \widehat{\sigma}_n + \log (1+p_n/n) + \log (\mu_n+\nu_n-\widehat{\alpha}_{nj} + m_{nj}) + o(\log n).
\end{align*}
By \ref{ass:mu}, we have
\begin{align*}
\mu_n - \widehat{\alpha}_{nj} + & \log (\mu_n+\nu_n-\widehat{\alpha}_{nj} + m_{nj})\\
&\le (R-r-c)\log n - \log \widehat{\sigma}_n - \frac{1}{2} \log (1+p_n/n) + o(\log n).
\end{align*}
Note that $\nu_n + m_{nj} = o(\log n)$ by \ref{ass:nu} and part (b). Since $c<R-r$ and by Lemma \ref{lemma:sigma}, we have
\begin{align} \label{core}
\mu_n - \widehat{\alpha}_{nj} \le (R-r-c)\log n - \log \widehat{\sigma}_n - \frac{1}{2} \log (1+p_n/n) + o(\log n),
\end{align}
and therefore
\begin{align*}
\widehat{\alpha}_{nj} \ge (r+c)\log n + \log \widehat{\sigma}_n + \log (1+p_n/n) + o(\log n).
\end{align*}
This implies
\begin{align*}
\widehat{\xi}_{nj} > C_2 \widehat{\sigma}_n^2 (1+p_n/n) n^{r+c-\zeta}.
\end{align*}
On the other hand, by \eqref{sol:alpha}, \eqref{core}, and \ref{ass:nu}, we have
\begin{align*}
|\widehat{\beta}_{nj}| \widehat{\xi}_{nj} = \widehat{\sigma}_n^2 \frac{\mu_n+\nu_n-\widehat{\alpha}_{nj} + m_{nj}}{\nu_n} \le C_3\widehat{\sigma}_n^2 (1+p_n/n) n^r.
\end{align*}

\item The arguments are in parallel to those in part (d).

\end{enumerate}
\end{proof}


\begin{lem} \label{lemma:xibound}
Suppose $p_n \times 1$ vectors $\underline{\balpha}_n$ and $\overline{\balpha}_n$ satisfy, given $\widehat{\bbeta}_n$ and $\widehat{\sigma}_n$,
\begin{align}
\label{def:loweralpha} |\widehat{\beta}_{nj}| e^{\underline{\alpha}_{nj}} &= \widehat{\sigma}_n \frac{\mu_n+\nu_n-\underline{\alpha}_{nj}-M_n}{\nu_n}, \qquad 1 \le j \le p_n,\\
\label{def:upperalpha} |\widehat{\beta}_{nj}| e^{\overline{\alpha}_{nj}} &= \widehat{\sigma}_n \frac{\mu_n+\nu_n-\overline{\alpha}_{nj}+M_n}{\nu_n}, \qquad 1 \le j \le p_n.
\end{align}
Let $\underline{\xi}_{nj} = \widehat{\sigma}_n e^{\underline{\alpha}_{nj}}$ and $\overline{\xi}_{nj} = \widehat{\sigma}_n e^{\overline{\alpha}_{nj}}$. Then, the following statements are true.

\begin{enumerate}
\item $\underline{\alpha}_{nj} \le \widehat{\alpha}_{nj} \le \overline{\alpha}_{nj}$ for all $1 \le j \le p_n$.

\item $\underline{\alpha}_{nj}$ is a decreasing function of $|\widehat{\beta}_{nj}|$ and $|\widehat{\beta}_{nj}| \underline{\xi}_{nj}$ is a decreasing function of $\underline{\alpha}_{nj}$. These hold for $\overline{\alpha}_{nj}$ and $\overline{\xi}_{nj}$ analogously.

\item Lemma \ref{lemma:xi}(c), \ref{lemma:xi}(d), and \ref{lemma:xi}(e) hold with $\widehat{\xi}_{nj}$ replaced by $\underline{\xi}_{nj}$ (or $\overline{\xi}_{nj}$) as well.
\end{enumerate}
\end{lem}
\begin{proof}
\begin{enumerate}
\item Obvious from \eqref{def:loweralpha}, \eqref{def:upperalpha}, and the definition of $M_n$.
\item Obvious from the definitions \eqref{def:loweralpha} and \eqref{def:upperalpha}.
\item The same arguments in the proof of \ref{lemma:xi}(c), \ref{lemma:xi}(d), and \ref{lemma:xi}(e) are valid with $m_{nj}$ and $\widehat{\alpha}_{nj}$ replaced by $M_n$ and $\underline{\alpha}_{nj}$ (or $\overline{\alpha}_{nj}$), respectively.
\end{enumerate}
\end{proof}


\begin{lem} \label{lemma:main1}
$P(\cA_n \nsubseteq \cA_0) = P(\cC_n \neq \emptyset) \rightarrow 0$.
\end{lem}
\begin{proof}
Suppose $\cC_n \neq \emptyset$. By \eqref{solution:beta} and \ref{ass:rho}, note that
\begin{align*}
\widehat{\sigma}_n \sum_{j \in \cC_n} e^{\widehat{\alpha}_{nj}} |\widehat{\beta}_{nj}| = {\widehat{\bxi}_{\cC_n}}' S_{n\cC\cC} \widehat{\bbeta}_{\cC_n} = O_p(h_n^{1/2}) - h_n,
\end{align*}
where $h_n = {\widehat{\bxi}_{\cC_n}}' S_{n\cC\cC} (X_{\cC_n}'X_{\cC_n})^{-1} S_{n\cC\cC} {\widehat{\bxi}_{\cC_n}}$.
We claim that $h_n \rightarrow_p \infty$ and $\textrm{RHS} \rightarrow_p -\infty$ while LHS stays positive, which yields $P(\cC_n \neq \emptyset) \rightarrow 0$.

Suppose $\max_{j \in \cC_n} |\widehat{\bbeta}_{nj} | \le C_2 n^{-1/2+(r+z-1)/2}$ for $\exists C_2>0$ and large $n$.
By Lemma \ref{lemma:sigma}, \ref{ass:sigma}, Lemma \ref{lemma:xibound}(b), Lemma \ref{lemma:xibound}(c) with $\zeta = (r+z-1)/4$, we have
\begin{align*}
\min_{j \in \cC_n} \underline{\xi}_{nj} > C_1 n^{1/2},
\end{align*}
for $\forall C_1>0$ and large $n$. By Lemma \ref{lemma:xibound}(a), we have
\begin{align*}
P( \| \widehat{\bxi}_{\cC_n} \| \le C_1 n^{1/2} \; \& \; \max_{j \in \cC_n} |\widehat{\bbeta}_{nj} | \le C_2 n^{-1/2+(r+z-1)/2} ) \rightarrow 0,
\end{align*}
for $\forall C_1>0$ and $\forall C_2>0$.

On the other hand, suppose $\| \widehat{\bxi}_{\cC_n} \| \le C_1 n^{1/2}$ for $\exists C_1>0$ and large $n$.
By the fact that the errors are gaussian and by \eqref{solution:beta}, we have $\max_{j \in \cC_n} |\widehat{\bbeta}_{nj} | = o_p(n^{-1/2+\zeta})$ for $\forall \zeta>0$. Therefore, we have
\begin{align*}
P( \| \widehat{\bxi}_{\cC_n} \| \le C_1 n^{1/2} \; \& \; \max_{j \in \cC_n} |\widehat{\bbeta}_{nj} | > C_2 n^{-1/2+(r+z-1)/2} ) \rightarrow 0.
\end{align*}

We have reached
\begin{align*}
P( \| \widehat{\bxi}_{\cC_n} \| \le C_1 n^{1/2}) \rightarrow 0,
\end{align*}
for $\forall C_1>0$. Since $h_n \ge (\tau_2n)^{-1} \| \widehat{\bxi}_{\cC_n} \|^2$, we have $h_n \rightarrow_p \infty$, as claimed.
\end{proof}


\begin{lem} \label{lemma:main2}
$P(\cA_0 \subsetneq \cA_n) = P(\cC_n = \emptyset \, \& \, \cD_n \neq \emptyset) \rightarrow 0$.
\end{lem}
\begin{proof}
Suppose $\cC_n = \emptyset$ and $\cD_n \neq \emptyset$. This implies that
\begin{align*}
|X_{\cD_n}' (\by_n - X_n \widehat{\bbeta}_n)| < \widehat{\bxi}_{\cD_n} = \widehat{\sigma}_n e^{\widehat{\balpha}_{\cD_n}}.
\end{align*}
We claim that this inequality is satisfied with probability tending to 0.
Note that $\|\widehat{\bxi}_{\cD_n}\| = o_p(n^{1-u/2})$ by Lemma \ref{lemma:sigma}, Lemma \ref{lemma:xi}(a), and \ref{ass:q}.
On the other hand, note from \eqref{solution:residual} that
\begin{align*}
X_{\cD_n}' (\by_n - X_n \widehat{\bbeta}_n) &= X_{\cD_n}' ( I - H_{\cB_n} ) X_{\cD_n} \bbeta_{0\cD_n} + X_{\cD_n}'( I - H_{\cB_n} ) \bepsilon_n\\
& \qquad + X_{\cD_n} X_{\cB_n} \left( X_{\cB_n}' X_{\cB_n} \right)^{-1} S_{n\cB\cB} \widehat{\bxi}_{\cB_n}\\
&= X_{\cD_n}' ( I - H_{\cB_n} ) X_{\cD_n} \bbeta_{0\cD_n} + O_p(n^{1/2}q_n^{1/2}) + o_p(n^R q_n^{1/2})\\
&= X_{\cD_n}' ( I - H_{\cB_n} ) X_{\cD_n} \bbeta_{0\cD_n} + o_p(n^{1-u/2}).
\end{align*}
Since $\|X_{\cD_n}' ( I - H_{\cB_n} ) X_{\cD_n} \bbeta_{0\cD_n}\| \ge C nq_n^{-1/2}$ for some constant $C$ by \ref{ass:beta} and \ref{ass:XX}, the claim follows.
\end{proof}


\begin{lem} \label{lemma:sigma2}
$\widehat{\sigma}_n^2 = \Theta_p((1+p_n/n)^{-1})$.
\end{lem}
\begin{proof}
We already have $\widehat{\sigma}_n^2 = O_p((1+p_n/n)^{-1})$ by Lemma \ref{lemma:sigma}(a).

Lemma \ref{lemma:main1} and \ref{lemma:main2} shows that $P(\cA_n \neq \cA_0) \rightarrow 0$. Assume $\cA_n = \cA_0$, then by \eqref{solution:SSE:full}, we have
\begin{align*}
\widehat{c}_{1n} \ge \frac{1}{2} \bepsilon_n' (I - H_{\cA_0}) \bepsilon_n = \frac{1}{2} (n-q_n) \sigma_0^2 + O_p(n^{1/2}).
\end{align*}
By the lower bound in \eqref{sigma_bound}, we have $\widehat{\sigma}_n^{-2} = O_p(1+p_n/n)$.
\end{proof}


\begin{proof} [Proof of Theorem \ref{thm:oracle}]
By virtue of Lemma \ref{lemma:main1} and \ref{lemma:main2}, we assume $\cA_n = \cA_0$, and note that, by \eqref{solution:beta}, we have
\begin{align}
\nonumber \widehat{\bbeta}_{n\cA_0} - \bbeta_{0\cA_0} &= \left( X_{\cA_0}' X_{\cA_0} \right)^{-1} X_{\cA_0}' \bepsilon_n - \left( X_{\cA_0}' X_{\cA_0} \right)^{-1} S_{n\cA\cA} \widehat{\bxi}_{\cA_0}\\
\label{diffbeta0} &= O_p(n^{(u-1)/2}) + o(n^{-u/2}) = o(n^{-u/2}),
\end{align}
by Lemma \ref{lemma:xi}(a) and \ref{ass:q}.
Since \ref{ass:beta}, this implies the sign consistency $P(\sign(\widehat{\bbeta}_{nj}) = \sign(\bbeta_{0j}), \forall j) \rightarrow 1$.

\eqref{diffbeta0} also implies that $\max_{j \in \cA_0} |\widehat{\beta}_{nj}|^{-1} = O_p(q_n^{1/2}) = O_p(n^{u/2})$ by \ref{ass:beta}, and that $\max_{j \in \cA_0} \overline{\xi}_{nj} = O_p(n^{r+u/2})$ by Lemmas \ref{lemma:sigma}, \ref{lemma:xibound}(b), and \ref{lemma:xibound}(c).
Then, by Lemma \ref{lemma:xibound}(a), we have $\max_{j \in \cA_0} \widehat{\xi}_{nj} = O_p(n^{r+u/2})$.
We have reached
\begin{align} \label{diffbeta1}
\widehat{\bbeta}_{n\cA_0} - \bbeta_{0\cA_0} &= \left( X_{n\cA_0}' X_{n\cA_0} \right)^{-1} X_{n\cA_0}' \bepsilon_n - o_p(n^{-1/2}).
\end{align}
This proves the asymptotic normality in part (b).

For the existance of the solution, we verify the KKT condition \eqref{sol:inactive}
\begin{align*}
|\bx_{nj}' (\by_n - X_n \widehat{\bbeta}_n)| < \widehat{\bxi}_{nj}, \qquad j \notin \cA_0.
\end{align*}
Note that, by \eqref{sol:beta1} and \eqref{diffbeta1},
\begin{align*}
X_{n\cA_0^c}' (\by_n - X_n \widehat{\bbeta}_n) &= X_{n\cA_0^c}' ( I - H_{n\cA_0} ) \bepsilon_n + o(n^{1/2} \rho_n).
\end{align*}
Since $\bx_{nj}'\bx_{nj} = n$ for all $j$ and $e_{ni}$ are gaussian, we have
\begin{align*}
\max_{j \notin \cA_0} |\bx_{nj}' ( I - H_{n\cA_0} ) \bepsilon_n| = O_p(n^{1/2} (\log p_n)^{1/2}) = O_p(n^{(U+1)/2}).
\end{align*}
Therefore, we have
\begin{align*}
\max_{j \notin \cA_0} |\bx_{nj}' (\by_n - X_n \widehat{\bbeta}_n)| = O_p(n^{(U+1)/2}).
\end{align*}
On the other hand, note that $\max_{j \notin \cA_0} |\widehat{\beta}_{nj}| = 0$. By Lemmas \ref{lemma:xibound}(b), \ref{lemma:xibound}(c), and \ref{lemma:sigma2}, we have $\max_{j \notin \cA_0} \underline{\xi}_{nj}^{-1} = O_p(n^{-R+\zeta})$ for any $\zeta > 0$. By Lemma \ref{lemma:xibound}(a), we have $\max_{j \notin \cA_0} \widehat{\xi}_{nj}^{-1} = O_p(n^{-R+\zeta})$ for any $\zeta > 0$.
Since \ref{ass:mu}, by choosing $\zeta = (2R-U-1)/4$, we have
\begin{align*}
\max_{j \notin \cA_0} |\bx_{nj}' (\by_n - X_n \widehat{\bbeta}_n)| < \min_{j \notin \cA_0} |\widehat{\xi}_{nj}|
\end{align*}
with probability tending to 1.

For uniqueness, note that from \eqref{eq:marginal}
\begin{align*}
-\frac{\partial^2 \log \pi(\btheta|\by_n,X_n)}{\partial \balpha \partial \balpha'}  = \frac{1}{\nu_n} I + W D_{\bkappa} W' + \frac{1}{\sigma} D_{e^{\balpha}} D_{|\bbeta|},
\end{align*}
where $W$ is a $p$ by $p(p-1)/2$ matrix, whose column $\bw_{jk}$ corresponding to the edge $(j,k)$ is $\be_j-\be_k$, and
\begin{align*}
\kappa_{jk}  = G_{n,jk} a_{\omega n} \frac{4\nu_n b_{\omega n} - 2(\alpha_{nj}-\alpha_{nk})^2}{(2\nu_n b_{\omega n} + (\alpha_{nj} -\alpha_{nk})^2)^2}.
\end{align*}
Since $|\kappa_{jk}| \le \frac{a_{\omega n}}{\nu_n b_{\omega n}}$, we have $\left| \sum_{k \sim j} \kappa_{jk} \right| < \frac{L_n a_{\omega n}}{\nu_n b_{\omega n}}$.
By \ref{ass:omega}, we have
\begin{align*}
-\frac{\partial^2 \log \pi(\btheta|\by_n,X_n)}{\partial \balpha \partial \balpha'}  = \frac{1}{\nu_n} I + o(\frac{1}{\nu_n}) + \frac{1}{\sigma} D_{e^{\balpha}} D_{|\bbeta|}.
\end{align*}

The following table shows the Hessian matrix of $-\pi(\btheta|\by_n,X_n)$ with respect to $\sigma$, $\bbeta_{\cA_0}$, $\balpha_{\cA_0}$, and $\balpha_{\cA_0^c}$.
\begin{center}
\begin{tabular}{c|cccc}
w.r.t. & $\sigma$ & $\bbeta_{\cA_0}$ & $\balpha_{\cA_0}$ & $\balpha_{\cA_0^c}$\\
\hline
$\sigma$ & $ \frac{6c_{1n}}{\sigma_n^4} + \frac{2c_{2n}}{\sigma_n^3} - \frac{c_{3n}}{\sigma_n^2}$ &  \\
$\bbeta_{\cA_0}$ & $\frac{2X_{n\cA_0}'(\by_n-X_n\bbeta_n)}{\sigma_n^3} - \frac{ S_{\cA_0}  e^{\balpha_{\cA_0}} }{\sigma_n^2}$ & $\frac{X_{n\cA_0}'X_{n\cA_0}}{\sigma_n^2}$ \\
$\balpha_{\cA_0}$ & $-\frac{ D_{|\bbeta_{\cA_0}|} e^{\balpha_{\cA_0}} }{\sigma_n^2}$ & $\frac{D_{e^{\balpha_{\cA_0}}} S_{\cA_0}}{\sigma_n}$ & $\frac{I}{\nu_n} + o(\frac{1}{\nu_n}) + \frac{D_{e^{\balpha_{\cA_0}}}D_{|\bbeta_{\cA_0}|}}{\sigma_n}$  \\
$\balpha_{\cA_0^c}$ & 0 & 0 & 0 & $\frac{I}{\nu_n} + o(\frac{1}{\nu_n})$\\
\end{tabular}
\end{center}
The Hessian evaluated at the solution is given by
\begin{center}
\begin{tabular}{c|cccc}
w.r.t. & $\sigma$ & $\bbeta_{\cA_0}$ & $\balpha_{\cA_0}$ & $\balpha_{\cA_0^c}$\\
\hline
$\sigma$ &  {\color{red} $ \frac{4\widehat{c}_{1n}}{\widehat{\sigma}_n^4}$}  $+$  {\color{blue} $\frac{\widehat{c}_{2n}}{\widehat{\sigma}_n^3}$} &  \\
$\bbeta_{\cA_0}$ & {\color{red} $\frac{X_{n\cA_0}'(\by_n-X_n\widehat{\bbeta}_n)}{\widehat{\sigma}_n^3}$} & {\color{red} $\frac{X_{n\cA_0}'X_{n\cA_0}}{2\widehat{\sigma}_n^2}$} $+$ {\color{green} $\frac{X_{n\cA_0}'X_{n\cA_0}}{2\widehat{\sigma}_n^2}$} \\
$\balpha_{\cA_0}$ & {\color{blue} $-\frac{ D_{|\widehat{\bbeta}_{\cA_0}|} e^{\widehat{\balpha}_{\cA_0}} }{\widehat{\sigma}_n^2}$} & {\color{green} $\frac{D_{e^{\widehat{\balpha}_{\cA_0}}} S_{n\cA_0\cA_0}}{\widehat{\sigma}_n}$} & {\color{green} $\frac{I}{\nu_n} + o(\frac{1}{\nu_n})$} $+$ {\color{blue} $\frac{D_{e^{\widehat{\balpha}_{\cA_0}}}D_{|\widehat{\bbeta}_{\cA_0}|}}{\widehat{\sigma}_n}$}  \\
$\balpha_{\cA_0^c}$ & 0 & 0 & 0 & $\frac{I}{\nu_n} + o(\frac{1}{\nu_n})$
\end{tabular}
\end{center}

The red-colored submatrix is strictly positive definite. The blue-colored submatrix is positive semi-definite. We claim that the green-colored submatrix is asymptotically strictly positive definite. Note that the smallest eigen value of $X_{n\cA_0}'X_{n\cA_0}$ is greater than or equal to $n\tau_1$ by \ref{ass:XX}. The smallest eigen value of $I/\nu_n$ is greater than $(1+p_n/n) n^{r-\zeta}$ for any $\zeta>0$ and large $n$ by \ref{ass:nu}. On the other hand, the largest eigenvalue of $D_{e^{\balpha_{\cA_0}}} S_{\cA_0}$ is of $O_p((1+p_n/n)^{1/2} n^{r+u/2})$ as discussed above. The claim follows by \ref{ass:nu}.

We have proved that the objective function is strictly convex in the region where the solutions can reside. Suppose we have two distinct solutions. This is only possible if there is at least one non-convex point in the segment between the two points, which cannot be the case because the objective function is strictly convex in that region. Therefore, the solution is unique and this completes the proof.

\end{proof}


\subsection*{Proof of Theorem \ref{oracle:fix}}

\begin{proof}[Proof of Theorem \ref{oracle:fix}]
Lemmas \ref{lemma:sigma}--\ref{lemma:xibound} hold with $U=0$, $u=0$, $1/2<R<1$, $0<r<R-1/2$, and $0 \le z < 1$. Since the partial orthogonality is not assumed in the fixed $p$ case, we take a different stratege to prove the rest. We first prove $P(\cD_n \neq \emptyset) \rightarrow 0$, then $\widehat{\sigma}_n^{-2} = O_p(1)$, and then show $P(\cC_n \neq \emptyset \, \& \, \cD_n = \emptyset ) \rightarrow 0$.

Suppose $\cD_n \neq \emptyset$. Note from \eqref{solution:residual:full} that
\begin{align*}
X_{\cD_n}' (\by_n - X_n \widehat{\bbeta}_n) &= X_{\cD_n}' ( I - H_{\cA_n} ) X_{\cD_n} \bbeta_{0\cD_n} + X_{\cD_n}'( I - H_{\cA_n} ) \bepsilon_n\\
& \qquad + X_{\cD_n} X_{\cA_n} \left( X_{\cA_n}' X_{\cA_n} \right)^{-1} S_{n\cA\cA} \widehat{\bxi}_{\cA_n}\\
&= X_{\cD_n}' ( I - H_{\cA_n} ) X_{\cD_n} \bbeta_{0\cD_n} + O_p(n^{1/2}) + o_p(n^R)\\
&= X_{\cD_n}' ( I - H_{\cA_n} ) X_{\cD_n} \bbeta_{0\cD_n} + o_p(n).
\end{align*}
Since $X_{\cD_n}' ( I - H_{\cA_n} ) X_{\cD_n} \bbeta_{0\cD_n} \ge C n$ in probability for $\exists C>0$ by \ref{ass:fix:beta} and \ref{ass:fix:XX}, the KKT condition \eqref{sol:inactive} cannot be satisfied due to Lemma \ref{lemma:xi}(a). This implies $P(\cD_n = \emptyset) \rightarrow 1$.

Now assume $\cD_n = \emptyset$, then by \eqref{solution:SSE:full}, we have
\begin{align*}
\widehat{c}_{1n} \ge \frac{1}{2} \bepsilon_n' (I - H_{\cA_n}) \bepsilon_n \ge \frac{1}{2} (n-p) \sigma_0^2 + O_p(n^{1/2}).
\end{align*}
By the lower bound in \eqref{sigma_bound}, we have $\widehat{\sigma}_n^{-2} = O_p(1)$.

Suppose $\cC_n \neq \emptyset$ and $\cD_n = \emptyset$. By \eqref{sol:beta1}, we have
\begin{align*}
\widehat{\sigma}_n \sum_{j =1}^p e^{\widehat{\alpha}_{nj}} |\widehat{\beta}_{nj}| = {\widehat{\bxi}_{\cA_n}}' S_{n\cA\cA} \widehat{\bbeta}_{\cA_n} = O_p(h_n^{1/2}) - h_n,
\end{align*}
where $h_n = {\widehat{\bxi}_{\cA_n}}' S_{n\cA\cA} (X_{\cA_n}'X_{\cA_n})^{-1} S_{n\cA\cA} {\widehat{\bxi}_{\cA_n}}$. By the similar argument as in Lemma \ref{lemma:main1}, we can show that $h_n \rightarrow_p \infty$ and $P(\cC_n \neq \emptyset \, \& \, \cD_n = \emptyset ) \rightarrow 0$. (Use the fact $\widehat{\sigma}_n^{-2} = O_p(1)$).

Now we have all the resuls that are analogous to Lemmas \ref{lemma:sigma}--\ref{lemma:sigma2}. The rest of the proof is analogous to the proof of Theorem \ref{thm:oracle}.
\end{proof}

\end{document}